\documentclass[journal]{IEEEtran}

\usepackage[compatibility=false, font=small]{caption}
\usepackage[font=small]{subcaption}
\usepackage{graphicx} 
\usepackage{array}
\usepackage{amsmath, amssymb}
\usepackage{xcolor}
\usepackage{mathtools} 
\usepackage{amsthm} 
\newtheorem{theorem}{Theorem}

\newtheorem{corollary}{Corollary}[theorem]
\usepackage{breqn}  
\usepackage{balance}

\usepackage{orcidlink} 

\DeclareMathOperator{\E}{\mathbb{E}}
\DeclareMathOperator{\F}{\mathcal{F}}
\DeclareMathOperator{\Fi}{\mathcal{F}^{\;-1}}
\DeclareMathOperator*{\argmax}{arg\,max}
\DeclareMathOperator{\Var}{Var}

\makeatletter
\DeclareRobustCommand{\pdot}{\mathbin{\mathpalette\pdot@\relax}}
\newcommand{\pdot@}[2]{%
  \ooalign{%
    $\m@th#1\odot$\cr
    \hidewidth$\m@th#1\cdot$\hidewidth\cr
  }%
}
\makeatother

\let\oldeqref\eqref
\renewcommand{\eqref}[1]{Eq.~\oldeqref{#1}}

\usepackage{cleveref}

\begin{document}

\title{Snapshot Synthetic Aperture Imaging\\with Boiling Speckle}

\author{
    Janith~B.~Senanayaka\textsuperscript{\orcidlink{0000-0003-1375-535X}}~and~Christopher~A.~Metzler\textsuperscript{\orcidlink{0000-0001-6827-7207}}
\thanks{This work was supported in part by ONR award no.~N00014-23-1-2752 and NSF CAREER Award no.~CCF2339616.}
\thanks{Janith B. Senanayaka is with the Department
of Electrical and Computer Engineering, University of Maryland, College Park, MD, 20740 USA e-mail: janith@umd.edu. Christopher A. Metzler is with the Department of Computer Science, University of Maryland, College Park, MD, 20740 USA e-mail: metzler@umd.edu.}
\thanks{This work has been submitted to the IEEE for possible publication. Copyright may be transferred without notice, after which this version may no longer be accessible.}}

\maketitle

\begin{abstract}
Light-based synthetic aperture (SA) imaging methods, such as Fourier Ptychography, have brought breakthrough high-resolution wide-field-of-view imaging capabilities to microscopy. 
While these technologies promise similar improvements in long-range imaging applications, macroscale light-based SA imaging is significantly more challenging. 
In this work, we first demonstrate that speckle noise is particularly problematic for light-based SA systems.
Specifically, we prove that it is fundamentally impossible to perform SA imaging of fully diffuse scenes if one captures sequential measurements that suffer from per-measurement-independent speckle. 
We then develop a snapshot SA imaging method and aperture-phase-synchronization strategy that can overcome this limitation and enable SA imaging. 
Remarkably, we further demonstrate, in simulation, that \emph{speckle can be exploited to recover missing spatial frequency information in SA imaging systems with distributed, non-overlapping apertures}. 
That is, one can use speckle to improve the resolution of an SA imaging system.
\end{abstract}


%
\IEEEpeerreviewmaketitle

\section{Introduction}
\IEEEPARstart{S}{ynthetic} aperture (SA) imaging systems synthesize high-resolution wide-field-of-view images by computationally integrating a sequence of low-resolution measurements captured with distinct illumination angles or sensor/aperture positions~\cite{gustafsson_surpassing_2000, zheng_wide-field_2013, holloway_toward_2016}. 
While SA imaging was originally introduced for radar and sonar~\cite{noauthor_synthetic_nodate}, it has since been translated into the optical domain~\cite{tippie_high-resolution_2011,zheng_wide-field_2013}.
Today light-based SA imaging techniques such as Fourier Ptychography (FP) have become a mainstay in the computational microscopy toolbox~\cite{konda_fourier_2020, zheng_concept_2021}. 
While light-based SA shows promise for long range imaging applications as well~\cite{dong_aperture-scanning_2014, holloway_savi_2017, li_far-field_2023, zhang_200_2024}, it remains under-utilized in that domain---we attribute this gap to speckle noise. 

Unlike in microscopy, in long-range SA applications, many targets of interest are likely to be optically rough, and thus introduce speckle noise.
Moreover, these targets are unlikely to be perfectly stationary; this speckle noise will change over time. That is, it will ``boil''.

\begin{figure}[t]
  \centering
  \subfloat[Ground truth]{\includegraphics[width=0.33\linewidth]{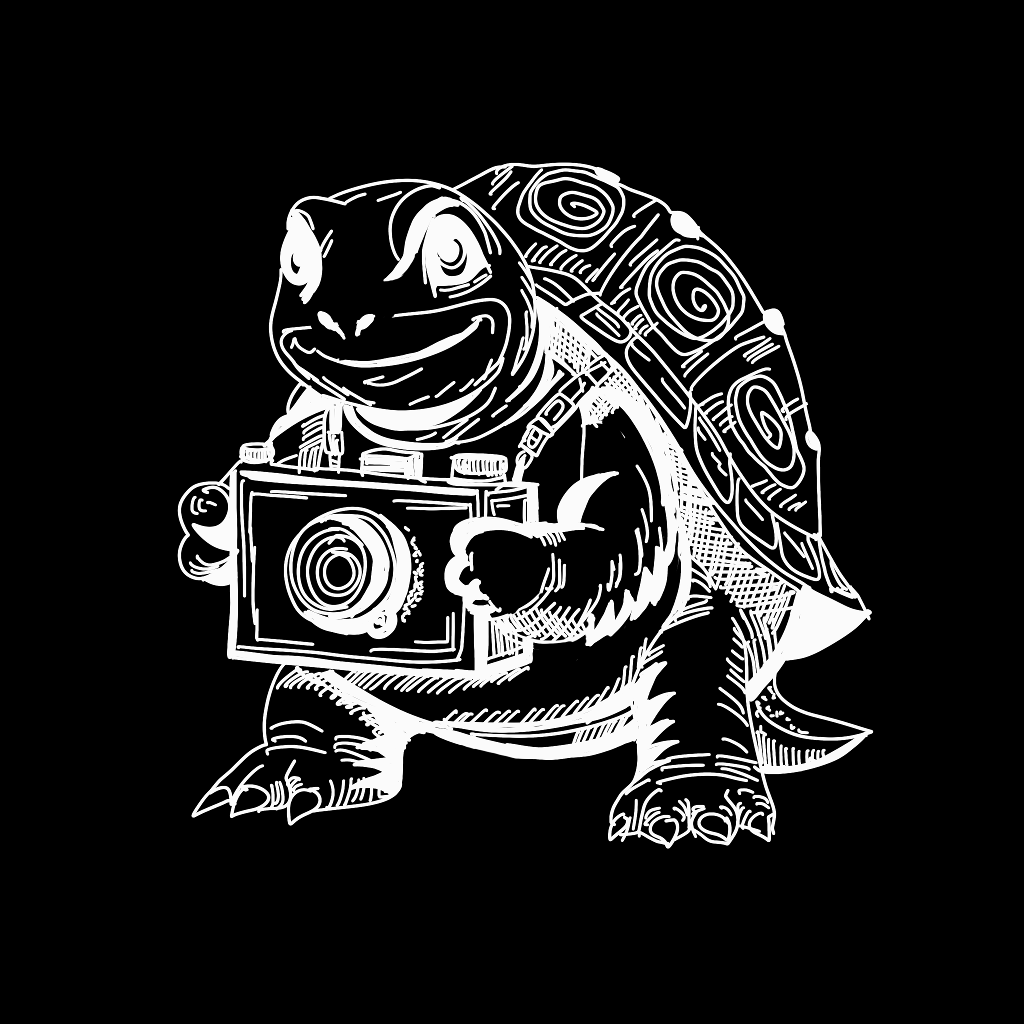}}
  \hfill
  \subfloat[Physical aperture]{\includegraphics[width=0.33\linewidth]{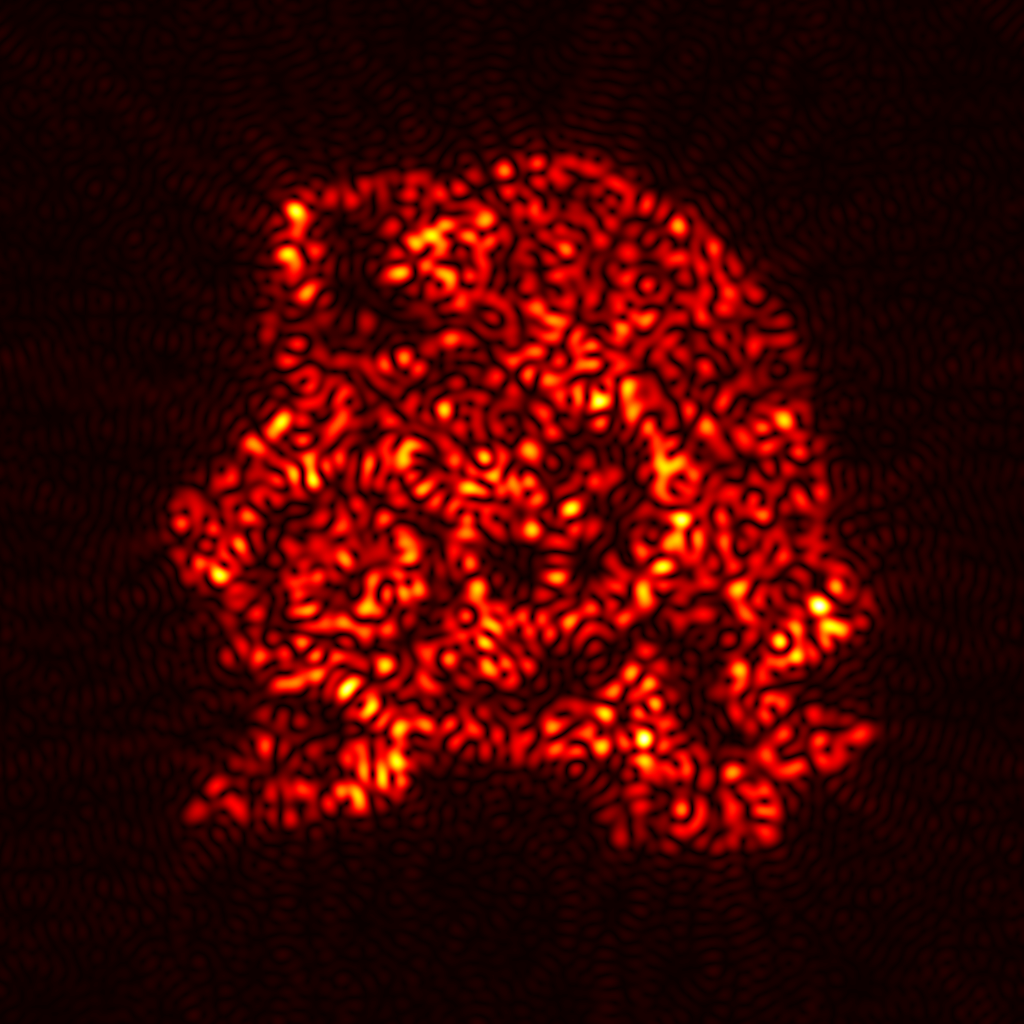}} 
  \hfill
  \subfloat[Synthetic aperture]{\includegraphics[width=0.33\linewidth]{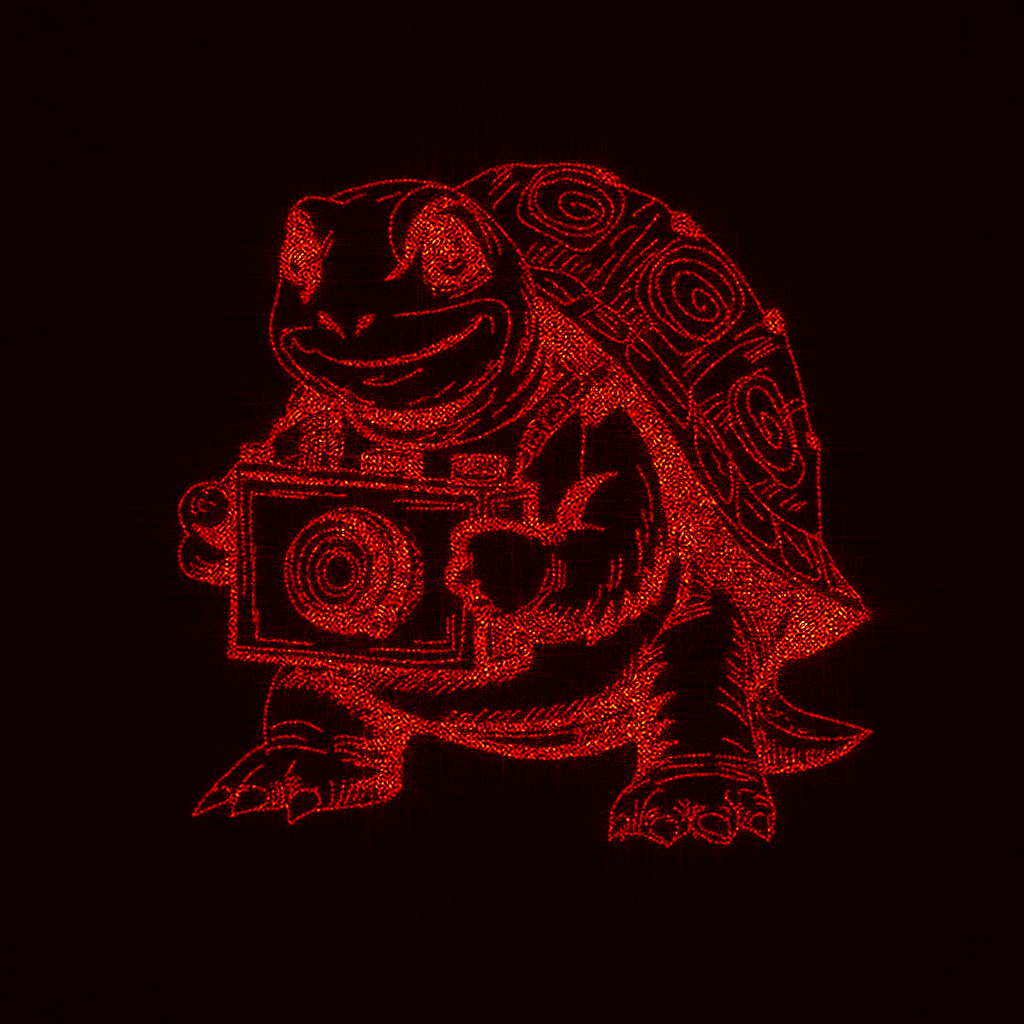}}  
  \caption{\textbf{Resolution improvement.} 
  Our SA imaging system can mitigate and exploit speckle noise to resolve fine, high-frequency details using measurements captured by a distributed array of sensors.
  }
  \label{fig:testudo}
\end{figure}

In this paper, we analyze SA imaging with boiling speckle. 
We prove that, with per-measurement-independent speckle and fully diffuse scenes, the distribution of far-field measurements is invariant to the aperture position.  
Consequently, \emph{light-based SA imaging is fundamentally impossible for fully diffuse scenes under per-measurement-independent speckle}. 

One can avoid independent speckle realizations by capturing measurements simultaneously in a snapshot with a sensor array~\cite{wang2022snapshot}. 
However, it is generally challenging, though not impossible~\cite{wang_metalens-based_2024}, to construct a snapshot imaging array with overlapping apertures. Without overlapping apertures, each subaperture image is subject to its own unknown global phase offset---one must compensate for this phase offset in order to perform SA imaging. 
Inspired by recent advances in wavefront shaping~\cite{yeminy_guidestar-free_2021, haim_image-guided_2025}, we propose a novel computational  aperture-phase-synchronization technique to estimate these relative phase offsets. 

While the proposed snapshot strategy improves resolution, the reconstructed images are still severely degraded by speckle noise. 
To mitigate speckle noise, we introduce a novel speckle averaging strategy.
We show this speckle averaging strategy both removes noise and \emph{recovers spectral information that lies between far-field subapertures}---and would not otherwise be recorded. 

In summary, our contributions are as follows:
\begin{itemize}
    \item We prove that boiling speckle makes SA imaging with sequentially captured measurements impossible.
    \item We introduce an aperture-phase-synchronization technique to enable snapshot SA imaging with boiling speckle and distributed, non-overlapping apertures.
    \item We demonstrate that speckle can be exploited to enable hallucination-free SA imaging with non-overlapping apertures.  
\end{itemize}

This manuscript significantly expands on a preliminary conference submission~\cite{11091570}. 
The conference submission contained only the translation invariance and impossibility result in~\Cref{sec:impossibility}. 
The snapshot imaging framework, aperture-phase-synchronization strategy, and speckle averaging strategy introduced in this manuscript have not appeared elsewhere.

\subsection*{What About Existing Long-Range SA Imaging Systems?}

Our central theoretical result suggests that sequentially captured long-range light-based SA imaging is impossible---which seems to be rebutted by several successful long-range light-based SA imaging demonstrations~\cite{holloway_savi_2017, li_far-field_2023, zhang_200_2024}. 
We attribute this mismatch between theory and practice to two pessimistic assumptions we make in our analysis: (a) that speckle fully boils (is independent) measurement-to-measurement and (b) that the objects of interest are fully diffuse with no specular reflections. 
If either assumption is broken, reconstruction with sequential measurements is still possible to some degree.

\section{Problem Definition and Paper Outline}\label{sec:problem-definition}

\begin{figure}[t]
    \centering
    \includegraphics[width=0.5\textwidth]{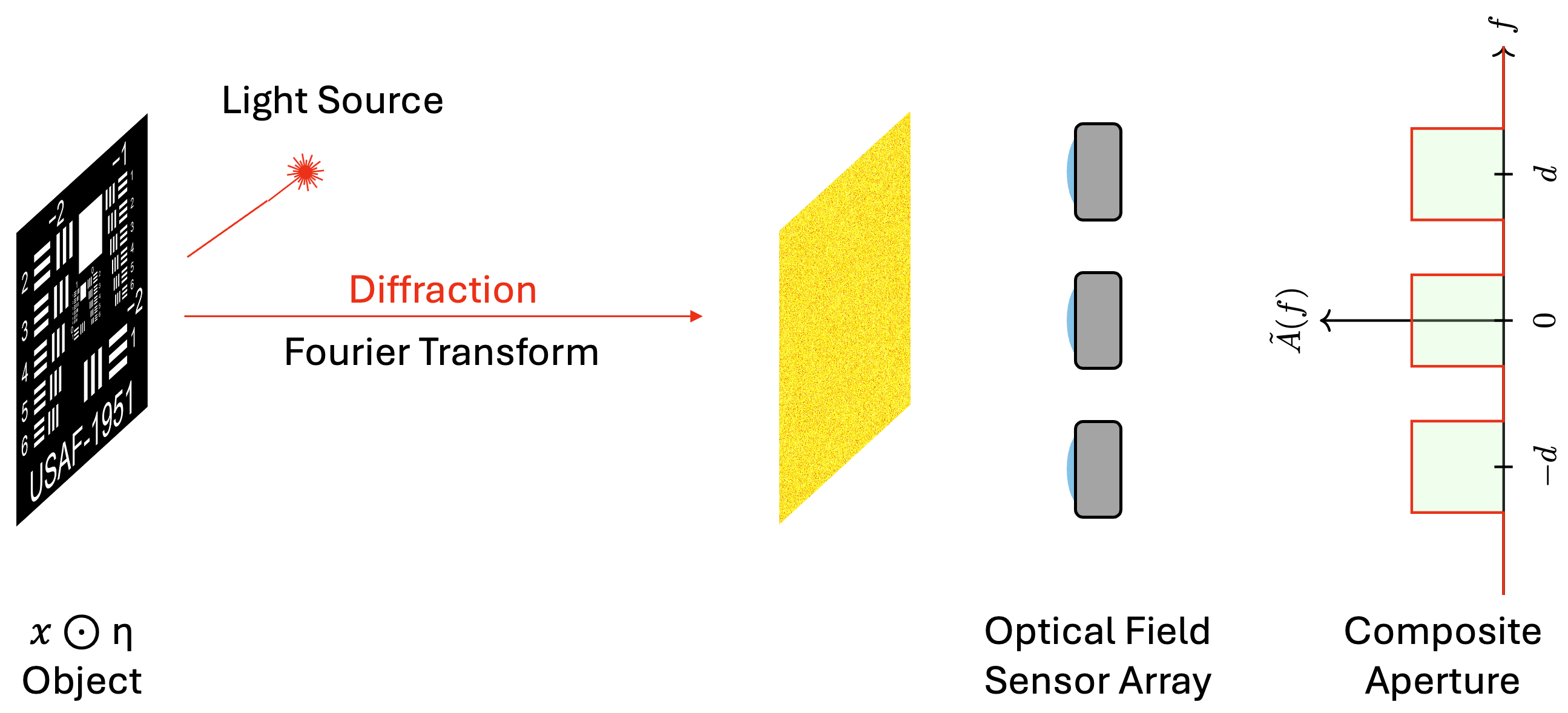}
    \caption{\textbf{Imaging setup.} The sensor array captures far-field (Fourier-plane) optical field measurements of the scene. The physical location of the field sensors determines the system's passband.
    }
    \label{fig:three-camera-setup}
\end{figure}

We consider a distributed aperture SA imaging system tasked with recovering the albedo $|x|^2$ of a coherently illuminated, fully Lambertian scene from a collection of non-overlapping diffraction-limited far-field measurements. 
We assume that each sensor records the complex-valued optical field at the aperture using holography or coded wavefront sensing~\cite{wu_wish_2019}. 
Thus we record a set of complex fields $U_{\ell,s}$ described by
\begin{equation}
    U_{\ell,s} =A_\ell \odot \F(x \odot \eta_s)e^{j{\phi}_{\ell,s}},
    \label{eq:problem-definition}
\end{equation}
where $\F$ is the 2D Fourier transformation operator (Fraunhofer propagation), $\odot$ denotes the Hadamard product (or element-wise product), $A_\ell$ denotes a binary mask describing the pupil function of the $\ell$\textsuperscript{th} aperture,  ${\phi}_{\ell,s}$ is an unknown phase offset associated with the measurement, and $\eta_s$ represents the $s$\textsuperscript{th} realization of circular Gaussian distributed speckle noise. 
Phase offsets ${\phi}_{\ell,s}$ are all but inevitable at optical frequencies---nm-scale manufacturing tolerances would be required to avoid them. 
Speckle noise is the result of the surface roughness of Lambertian scenes; it causes coherent light scattered off their surfaces to follow a circular Gaussian distribution~\cite{goodman_speckle_2020}. 

Such an imaging system is illustrated in Figure~\ref{fig:three-camera-setup}. 
The intensities of the Fourier transforms of measurements captured by each of such a system's apertures---i.e., image plane measurements---are illustrated in Figure~\ref{fig:comparison}. 
These images are equivalent to what one would observe with a Fourier pytchography imaging system.
As in Fourier ptychography, the position of the apertures determines the system's passband.

\subsection{Structure of the Paper}
This paper characterizes and improves the performance of such an imaging system. 
The paper is structured as follows:
In Section~\ref{sec:impossibility}, we will prove measurements following~\eqref{eq:problem-definition} are translation invariant and therefore SA imaging with sequentially captured measurements is impossible. 
Next in Section~\ref{sec:snapshot-imaging} we will show how translation invariance can be overcome by capturing the aperture measurements simultaneously. 
Moreover, in Section~\ref{sec:equalized-speckle-averaging} we will show that speckle can be exploited to improve the resolution of this system and in Section~\ref{sec:synchronization} we will introduce a variance-maximization strategy for synchronizing the phase offsets between the non-overlapping apertures' measurements. 
Section~\ref{sec:simulations} validates the proposed method in simulation.

\begin{figure}[t]
    \centering
    \renewcommand{\arraystretch}{0.3}
    \begin{tabular}{@{}m{0.05\linewidth}@{\hskip 0.005\linewidth}m{0.23\linewidth}@{\hskip 0.005\linewidth}m{0.23\linewidth}@{\hskip 0.005\linewidth}m{0.23\linewidth}@{\hskip 0.005\linewidth}m{0.23\linewidth}@{}}
        \rotatebox{90}{\small Aperture} &
        \includegraphics[width=\linewidth]{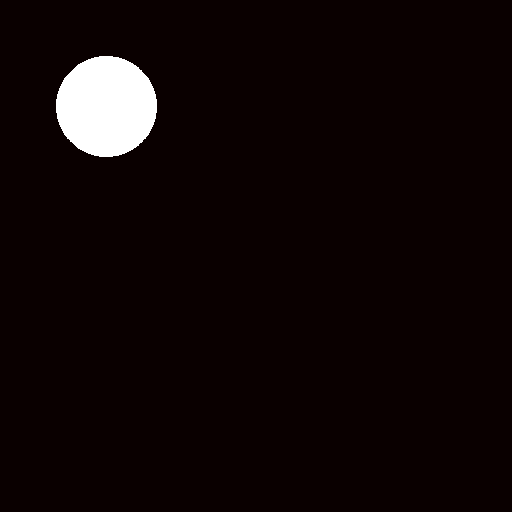} &
        \includegraphics[width=\linewidth]{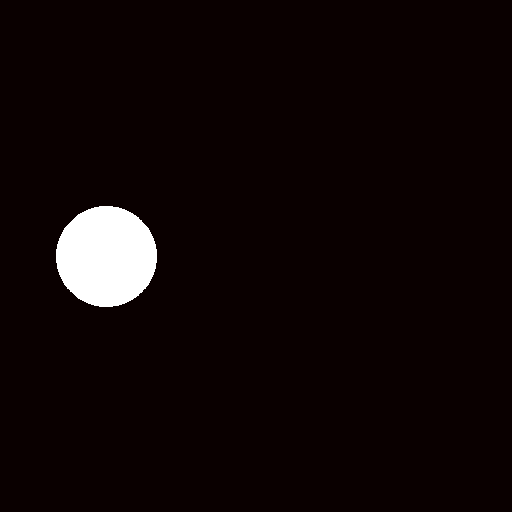} &
        \includegraphics[width=\linewidth]{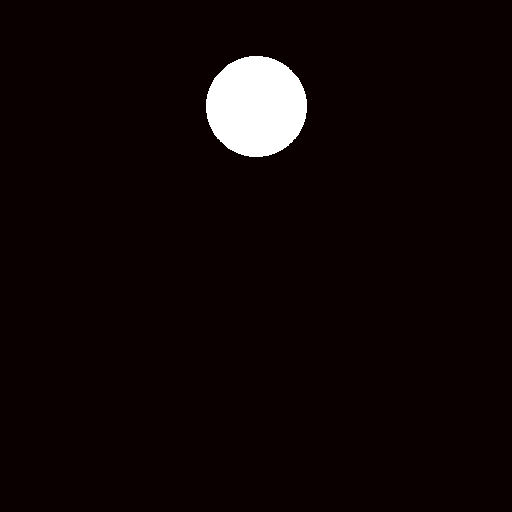} &
        \includegraphics[width=\linewidth]{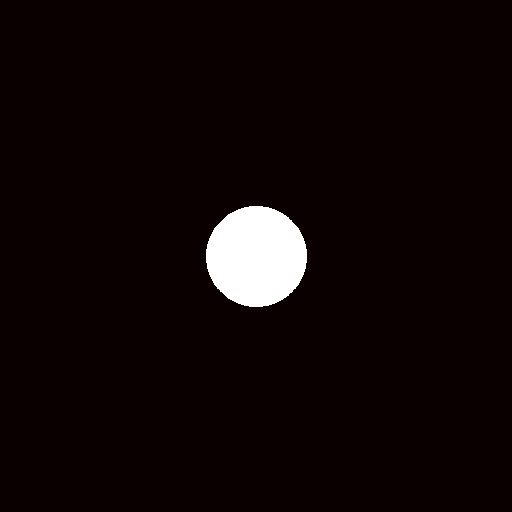} \\
    
        \rotatebox{90}{\small Specular} &
        \includegraphics[width=\linewidth]{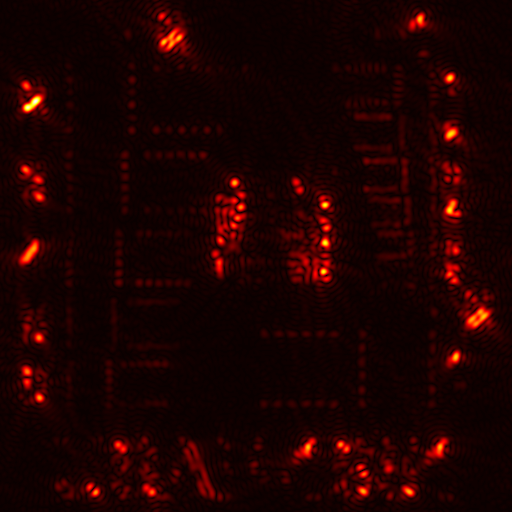} &
        \includegraphics[width=\linewidth]{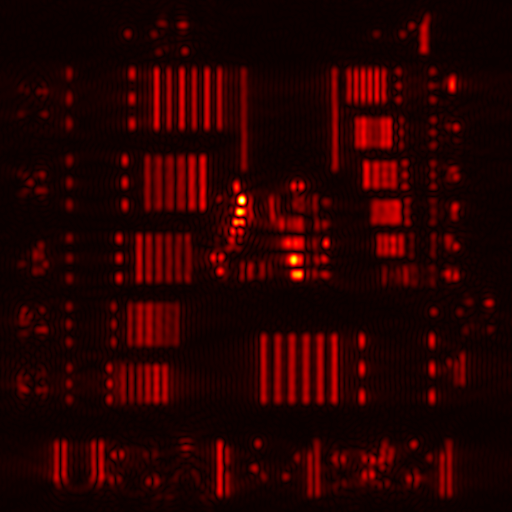} &
        \includegraphics[width=\linewidth]{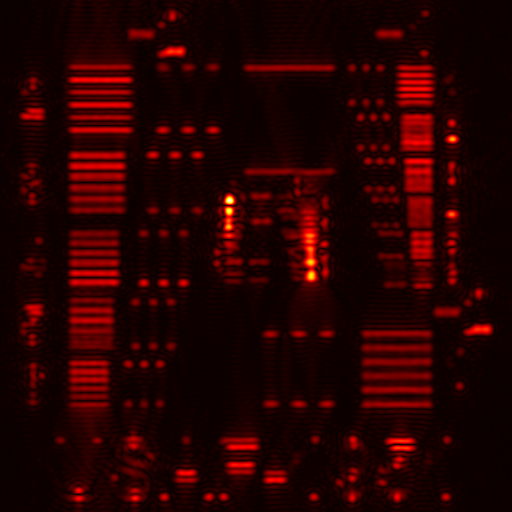} &
        \includegraphics[width=\linewidth]{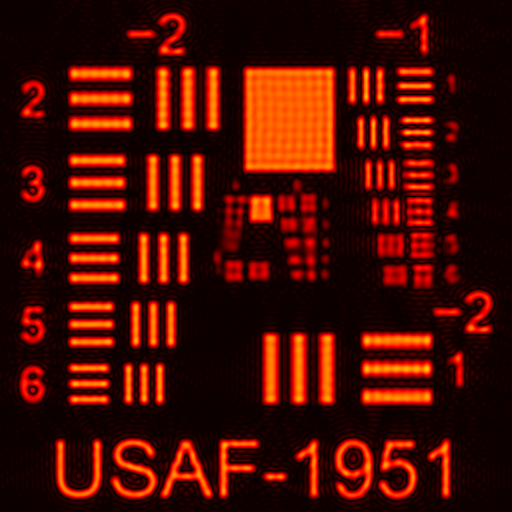} \\
    
        \rotatebox{90}{\small Diffuse} &
        \includegraphics[width=\linewidth]{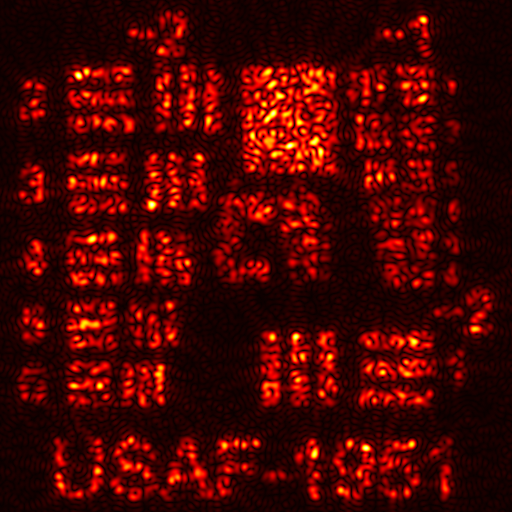} &
        \includegraphics[width=\linewidth]{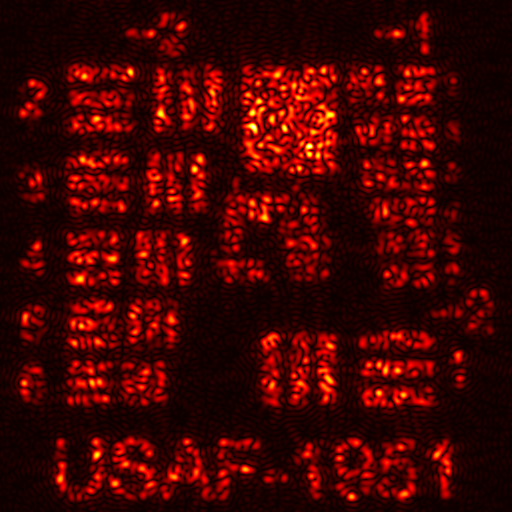} &
        \includegraphics[width=\linewidth]{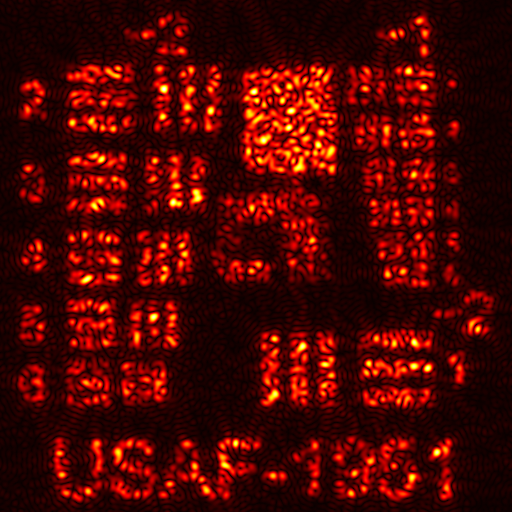} &
        \includegraphics[width=\linewidth]{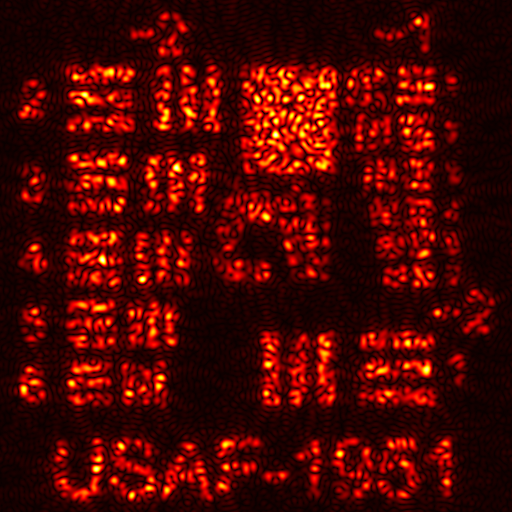} \\
    
        \rotatebox{90}{\small Average Diffuse} &
        \includegraphics[width=\linewidth]{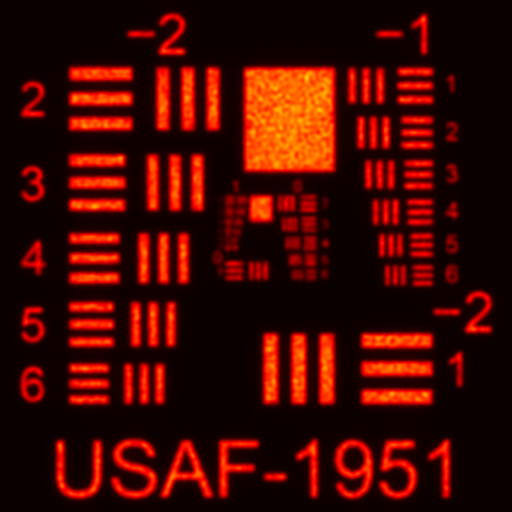} &
        \includegraphics[width=\linewidth]{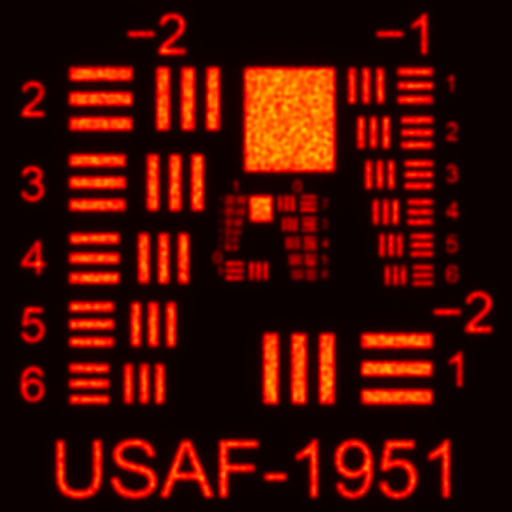} &
        \includegraphics[width=\linewidth]{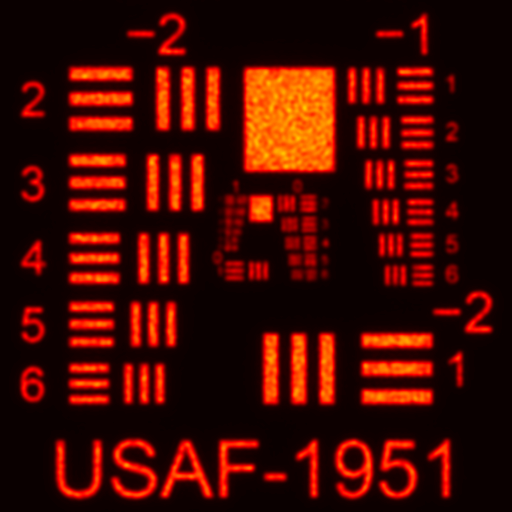} &
        \includegraphics[width=\linewidth]{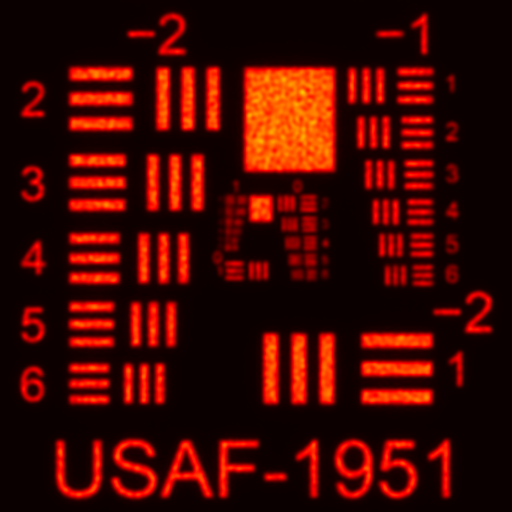} \\
    \end{tabular}

    \caption{\textbf{Comparison of specular and diffuse targets measured from different aperture positions.} 
    This figure compares image-plane intensity measurements of specular and diffuse targets as one varies the (Fourier-plane) aperture location.
    With specular targets, each aperture captures complementary frequency content about the scene.
    By contrast, the distribution of diffuse target measurements is translation invariant. Speckle averaging at each aperture location produces results in nearly identical (non-complementary) images.
    }
    \label{fig:comparison}
\end{figure}

\section{Related Work}\label{sec:related work}

\subsection{Light-Based Synthetic Aperture Imaging}
Light-based SA imaging is well established. 
In 2011, Tippie et al.~demonstrated one could perform light-based SA imaging by capturing overlapping holography (complex field) measurements~\cite{tippie_high-resolution_2011}. 
In 2013, Zheng et al.~introduced Fourier Ptychography (FP); a powerful SA technique that combines intensity-only measurements with phase retrieval algorithm to perform high resolution microscopy with a wide-field-of-view~\cite{ou_quantitative_2013}. 
Over the last decade, FP has become an important and widely used tool in microscopy \cite{ou_quantitative_2013, zheng_concept_2021, ou_embedded_2014, eadie_fourier_nodate, you_self-calibrating_2025}.

FP has since been extended to long-range imaging \cite{dong_aperture-scanning_2014}. 
Holloway et al. and Li et al.~demonstrated successful application of FP at ranges of several meters~\cite{holloway_savi_2017, li_far-field_2023}. 
Zhang et al.~recently demonstrated FP at 170m~\cite{zhang_200_2024}. 
We attribute this success, which is inconsistent with our theory, to correlations between measurements and/or specular components in the scenes.  

\subsection{Snapshot Imaging}

Unlike its counterpart in radar, measuring the complex optical field in light-based SA systems is challenging due to high optical frequencies. 
For this reason, FP rely on phase retrieval algorithms, which generally require significant overlap between the measurements' aperture positions~\cite{dong_sparsely_2014}. 
Typically, this aperture overlap is achieved by capturing measurements sequentially, however, several recent works have sought to capture these overlapping measurements in a snapshot.

Wang et al.~recently demonstrated one could use metalenses to capture overlapping aperture measurements in a single snapshot~\cite{wang_metalens-based_2024}. 
In parallel, Li et al.~introduced an illumination wavelength multiplexing approach to capture overlapping measurements in snapshot~\cite{li_snapshot_2024}. This approach assumes the target's albedo is constant across wavelengths.  
Alternatively, one can capture measurements with non-overlapping measurements and use deep-learning to inpaint the missing frequency content~\cite{wang2022snapshot, wang_learning-based_2023}.
Note that these methods do not measure this missing frequency content and so are prone to hallucination.

In contrast, our proposed method can perform SA imaging with conventional optical elements, without assumptions on the target's albedo, and without risk of hallucination. 

\subsection{Structured Illumination with Speckle}
In this work we exploit speckle to improve our systems resolution. Related, but distinct, ideas have been applied on several occasions.

Munson et al.~demonstrated that one can reconstruct the magnitude of a complex scene from frequency offset data, i.e., a small region in the Fourier transform offset from the origin, given that the phase is highly random~\cite{munson_image_1984}. 
Similarly, we reconstruct the magnitude of the complex scene with the highly random phase provided by speckle. 
However, they only use a fixed region in the Fourier transform whereas we combine multiple regions thus forming a synthetic aperture. 

Dong et al.~proposed illuminating a scene with a moving diffuser to achieve super-resolution by jointly estimating the scene and the unknown illumination pattern~\cite{dong_incoherent_2015}. 
Our problem differs substantially from this work: In their method, the unknown illumination pattern is constant up-to translation, and there are multiple measurements with that pattern, facilitating separation and estimation of the scene and the unknown illumination pattern. 
In contrast, we assume diffuse objects with ever-changing environmental conditions produce boiling speckle patterns, and each measurement undergoes a different unknown illumination pattern.

\subsection{Phase Synchronization \& Wavefront Shaping}

In this work we use variance maximization to synchronize the phase offsets between non-overlapping apertures. 
Our approach draws inspiration from a large body of related work.

Optimizing sharpness and related image quality metrics has long been used for digital refocusing and related problems~\cite{fienup2000synthetic}. 
In their SA holography work, Tippie et al.~synchronized {\em overlapping} aperture measurements using a two-step cross-correlation registration, where they use a sharpness metric for phase correction~\cite{tippie_high-resolution_2011}. 
In the context of wavefront shaping, Yeminy et al.~optimized the phase of a Spatial Light Modulator (SLM) by maximizing the variance of the reconstructed image~\cite{yeminy_guidestar-free_2021}.
Similarly Haim et al.~introduced an image-guided holographic wavefront shaping approach wherein they estimate virtual SLM patterns to improve the quality of a reconstructed image~\cite{haim_image-guided_2025}. 
To our knowledge, no previous works have used variance maximization to synchronize the phase offsets between {\em non-overlapping} aperture measurements.

\section{Translation Invariance \& Impossibility}\label{sec:impossibility}

In this section, we prove that the distribution of far-field measurements is independent of the aperture position, i.e., the distribution is translation invariant. 
From translation invariance it follows that light-based SA imaging is impossible for fully diffuse scenes with per-measurement-independent speckle.  

\begin{theorem}[Translation Invariance]
\label{thm:translation-invariance}
For a fully diffuse target, which produces fully developed speckle, and under per-measurement-independent and identically distributed speckle, the distribution of far-field measurements is invariant to the aperture translation. 
Mathematically
\begin{equation}
p(U_{\ell_1} | A_{\ell_1}) = p(U_{\ell_2} | A_{\ell_2}), \quad \forall \ell_1, \ell_2, \label{eq:translation-invaraince}
\end{equation}
where $U_{\ell}$ denotes far-field optical measurement at the image plane, and $A_\ell$ denotes the aperture at the $\ell$\textsuperscript{th} position.
\end{theorem}

\begin{proof}
     See Appendix~\ref{sec:appendix-a}.
\end{proof}
The proof uses the convolution property of Fourier transformation and the rotational invariance of circular Gaussian distribution, which maintains that multiplying by a phase factor does not change the statistics of the circular Gaussian.

From \Cref{thm:translation-invariance} it follows that  far-field image-plane intensity measurements (such as those captured in Fourier ptychography) are also translation invariant under per-measurement-independent speckle noise. 

This result is corroborated by Fig.~\ref{fig:comparison}, which compares specular and diffuse scene images captured from various aperture positions. 
As one averages the diffuse scene observations at a particular aperture location over many speckle realizations, each converges to the same image; their first moments are the same.

\begin{corollary}[Impossibility]
\label{corr:impossibility}
Light-based synthetic aperture imaging is impossible for fully diffuse scenes under per-measurement-independent speckle.
\end{corollary}

Consider the reconstruction of a high-resolution image of a scene using synthetic aperture imaging. 
Let us suppose that the scene is fully diffuse, meaning that the observed speckle in the image is fully developed. 
If sequential scanning is used to capture multiple images from different aperture positions, and the time difference between two exposures exceeds the de-correlation time of the speckle such that it makes the two speckle realizations independent of each other, we can conclude that all the images come from the same statistical distribution regardless of the aperture position. 
That is, the core assumption underlying SA imaging---that moving the aperture captures different frequency components---is invalid. 
As a result, light-based SA imaging is impossible in such a setting. 

\section{Snapshot Synthetic Aperture Imaging}\label{sec:snapshot-imaging}

To avoid speckle independence between measurements---and enable light based SA imaging---one can capture the subaperture measurements simultaneously in a snapshot with a sensor array, as illustrated in Fig.~\ref{fig:three-camera-setup}.  
The forward model associated with such an imaging system is described by~\eqref{eq:problem-definition}, which is reproduced here:
\begin{equation}
    U_{\ell,s} =A_\ell \odot \F(x \odot \eta_s)e^{j{\phi}_{\ell,s}}.\nonumber
    \label{eq:problem-definition_copy}
\end{equation}
Whereas with sequential measurements each subaperture observes a unique speckle realization, with a snapshot system all subapertures observe the same speckle realization. 
That is, for each speckle realization $\eta_s$, we have the full set of $L$ subaperture measurements $\{U_{\ell,s}\}_{\ell=1,...L}$. 
Note, however, that each aperture will still have its own unique global phase offset $\phi_{\ell,s}$.

\begin{figure}[t]
    \centering
    \begin{subfigure}[b]{0.495\linewidth}
        \includegraphics[width=\linewidth]{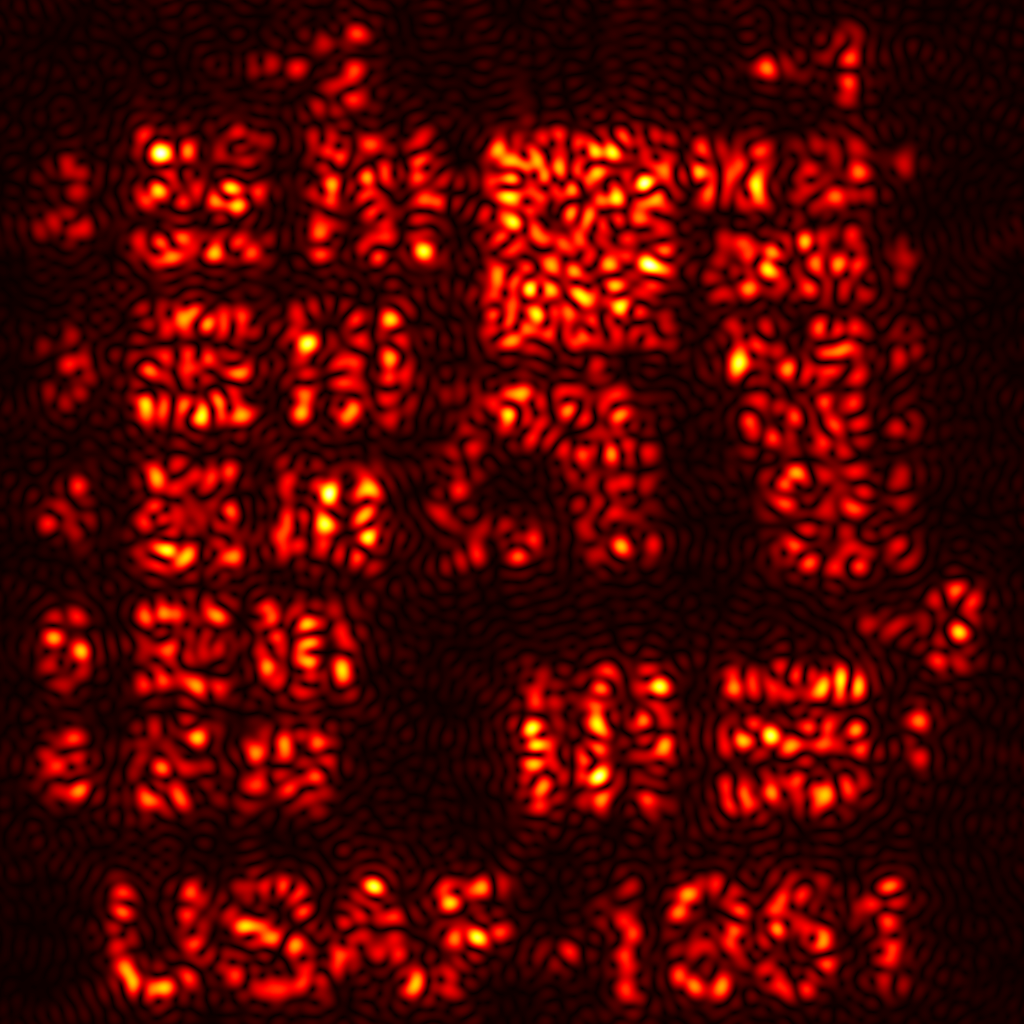}
        \caption{Center}\label{fig:intermediate-reconstruction-a}
    \end{subfigure}\hfill
    \begin{subfigure}[b]{0.495\linewidth}
        \includegraphics[width=\linewidth]{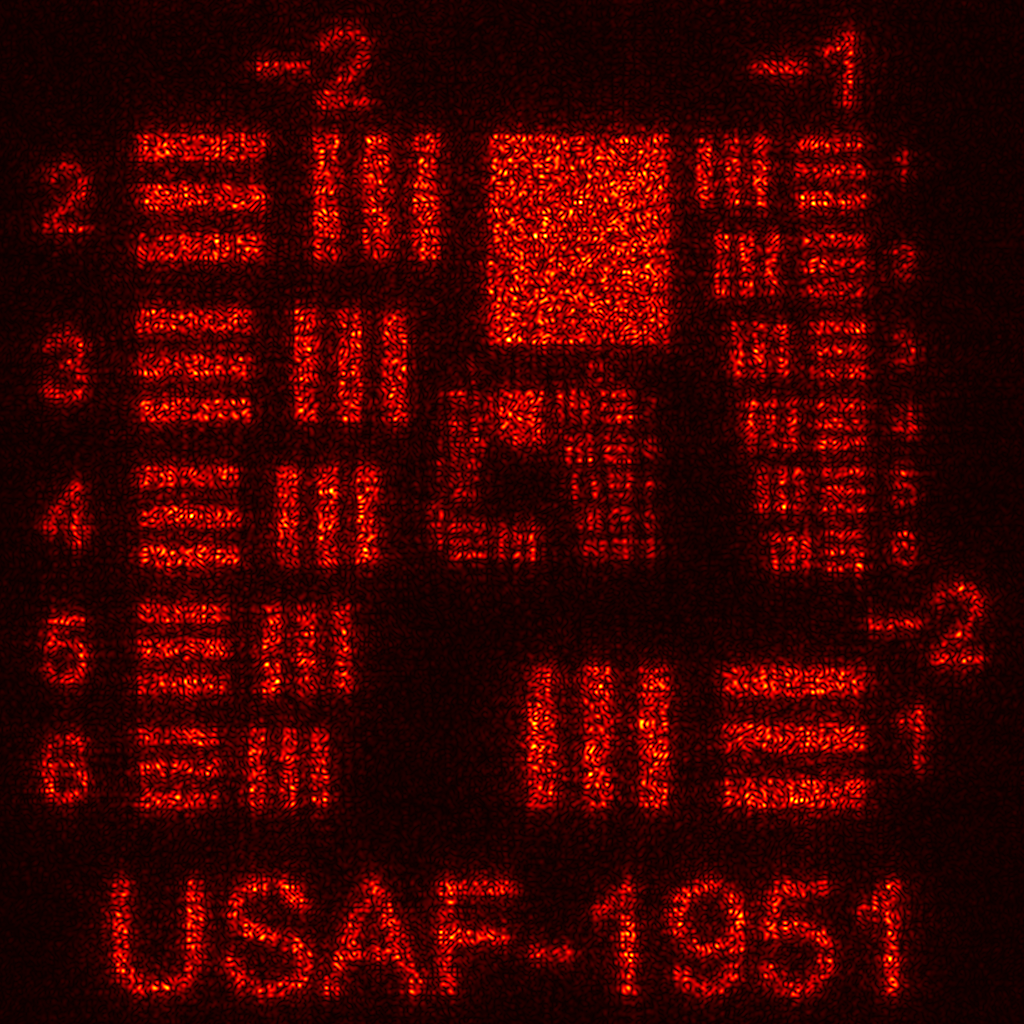}
        \caption{Composite}\label{fig:intermediate-reconstruction-b}
    \end{subfigure}
    \caption{\textbf{Comparison of single and composite aperture images.} 
    The composite aperture, formed by coherently combining fields captured from the snapshot optical field sensor, recovers high-frequency details.}
    \label{fig:intermediate-reconstruction}
\end{figure}

\begin{figure*}[t]
    \centering
    \begin{subfigure}[t]{0.7\textwidth}
        \centering
        \begin{subfigure}[t]{0.48\textwidth}
        \centering
            \includegraphics[width=\linewidth]{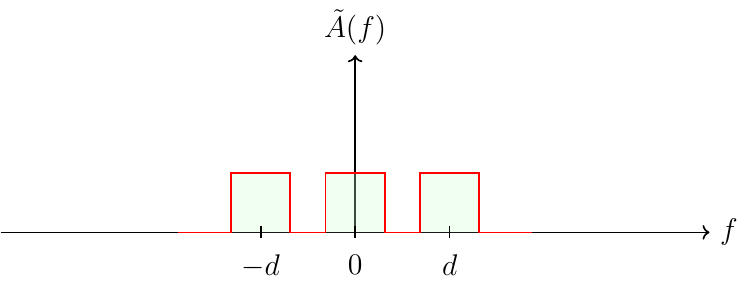}
            \caption*{(a) Composite Aperture}
        \end{subfigure}
        \begin{subfigure}[t]{0.48\textwidth}
        \centering
            \includegraphics[width=\linewidth]{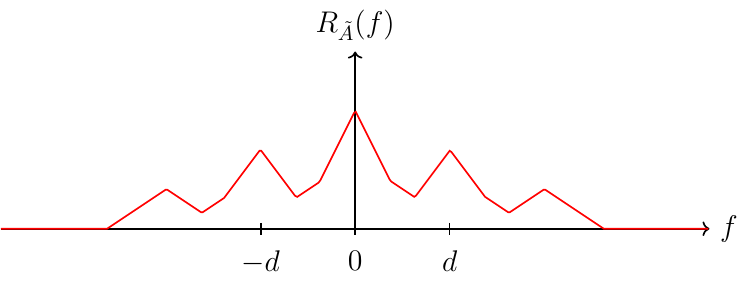}
            \caption*{(b) Autocorrelation Filter}
        \end{subfigure}
        \par\medskip
        \begin{subfigure}[t]{0.48\textwidth}
        \centering
            \includegraphics[width=\linewidth]{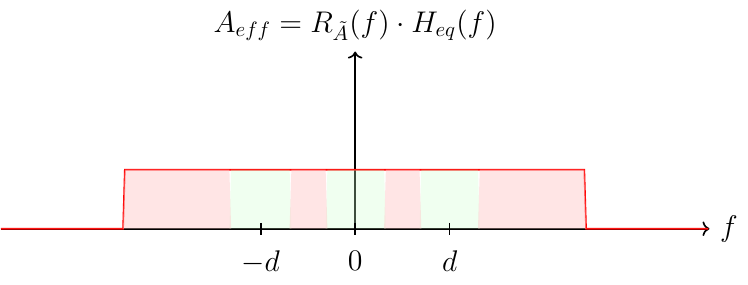}
            \caption*{(d) Effective Transfer Function}
            \phantomcaption\label{fig:theory-illustration-d}
        \end{subfigure}
        \begin{subfigure}[t]{0.48\textwidth}
        \centering
            \includegraphics[width=\linewidth]{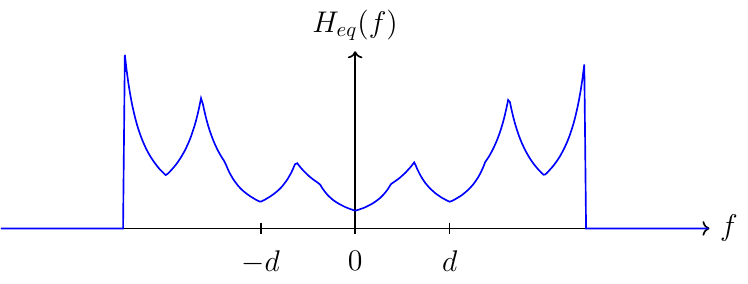}
            \caption*{(c) Equalization Filter}
        \end{subfigure}
    \end{subfigure}
    \hfill
    \begin{subfigure}[t]{0.25\textwidth}
        \centering
        \begin{subfigure}[t]{0.48\textwidth}
        \centering
            \includegraphics[width=\linewidth]{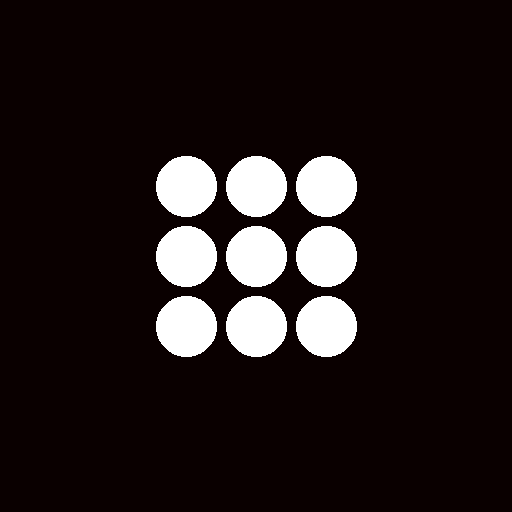}
            \caption*{(e) Composite}
        \end{subfigure}
        \begin{subfigure}[t]{0.48\textwidth}
        \centering
            \includegraphics[width=\linewidth]{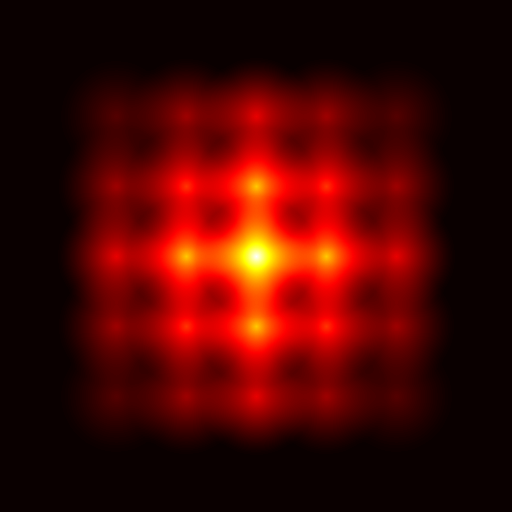}
            \caption*{(f) Autocorr.}
        \end{subfigure}
        \par\medskip
        \begin{subfigure}[t]{0.48\textwidth}
        \centering
            \includegraphics[width=\linewidth]{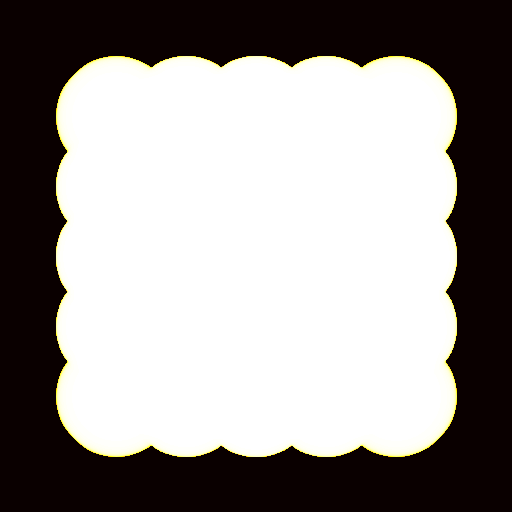}
            \caption*{(h) Effective}
        \end{subfigure}
        \begin{subfigure}[t]{0.48\textwidth}
        \centering
            \includegraphics[width=\linewidth]{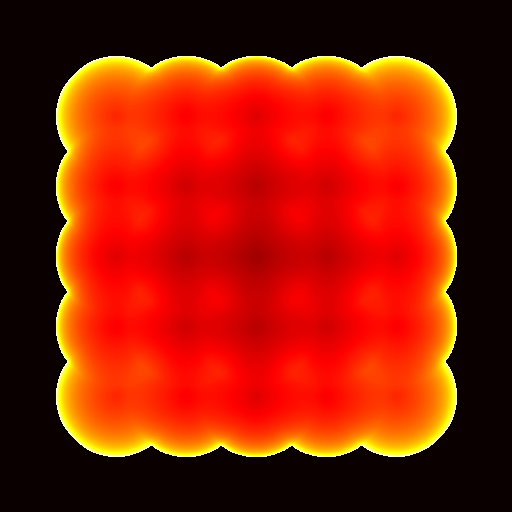}
            \caption*{(g) Equalizer}
        \end{subfigure}
    \end{subfigure}
    
    \caption{\textbf{Illustration of the underline transfer functions (or apertures) governing the proposed imaging process.} This figure shows the effect of averaging and equalization for both one-dimensional and two-dimensional signals/images. (a) composite aperture; (b) autocorrelation filter; (c) equalization filter; (d) effective transfer function; (e) composite aperture; (f) autocorrelation filter; (g) equalization filter; (h) effective transfer function.
    Incoherent averaging and equalization expands the composite aperture thus filling in the missing information between subapertures with true spectral information without hallucinations.}
    \label{fig:theory-illustration}
\end{figure*}

\begin{figure}[t]
    \centering
    \includegraphics[width=\linewidth]{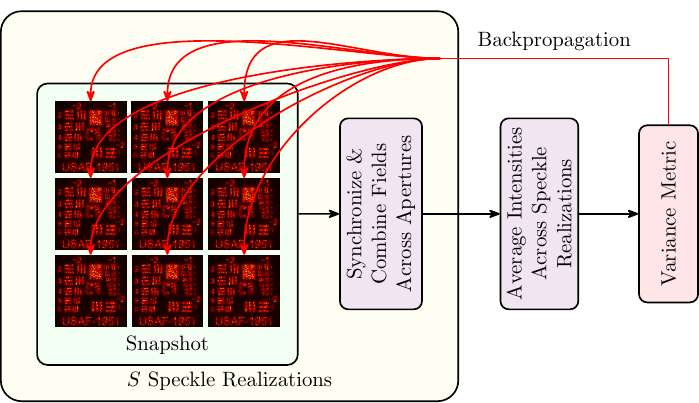}
    \caption{\textbf{Computational synchronization}. Unknown global phase offset estimation diagram. Backpropagation is used to update the global phase offsets in order to maximize the variance of the reconstructed image.
    }
    \label{fig:synchronization}
\end{figure}

\begin{figure}[t]
    \centering
    \begin{subfigure}[b]{0.495\linewidth}
        \includegraphics[width=\linewidth]{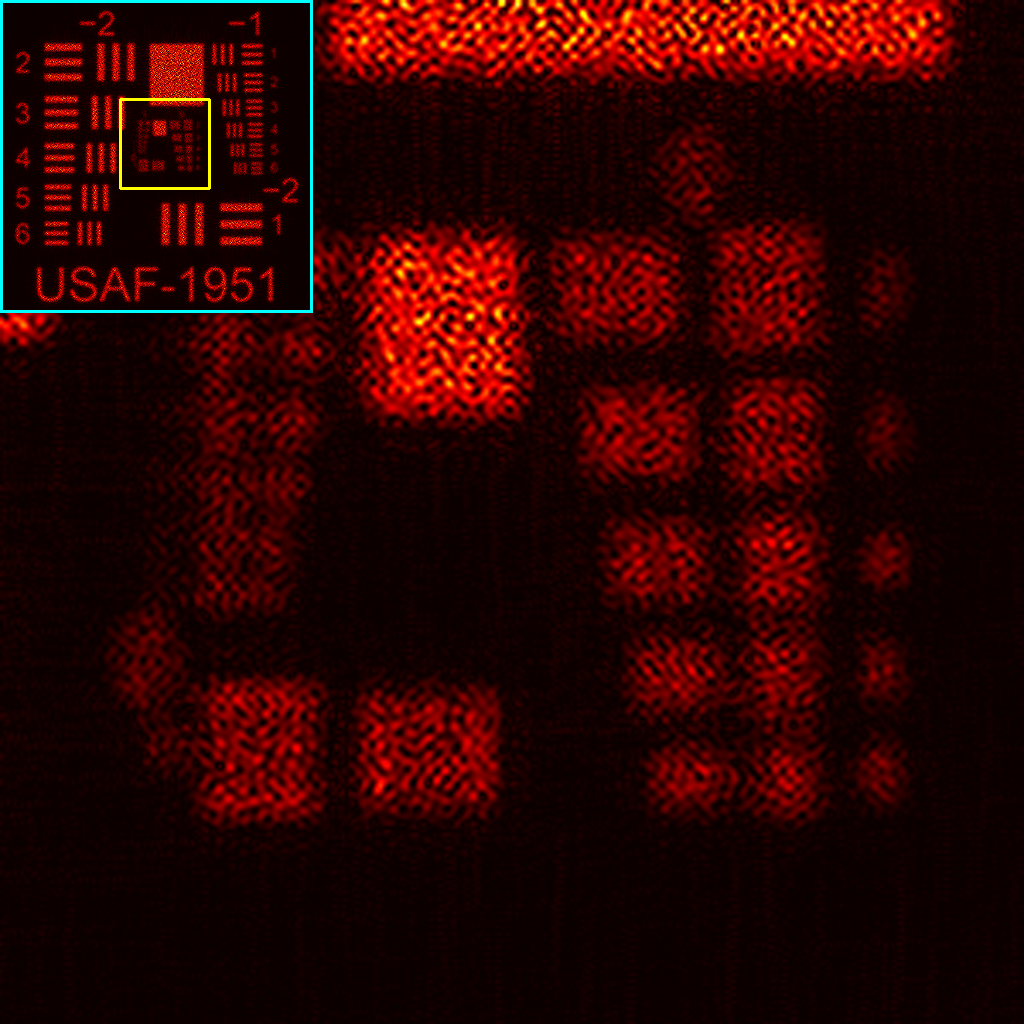}
        \caption{Without synchronization.}
    \end{subfigure}\hfill
    \begin{subfigure}[b]{0.495\linewidth}
        \includegraphics[width=\linewidth]{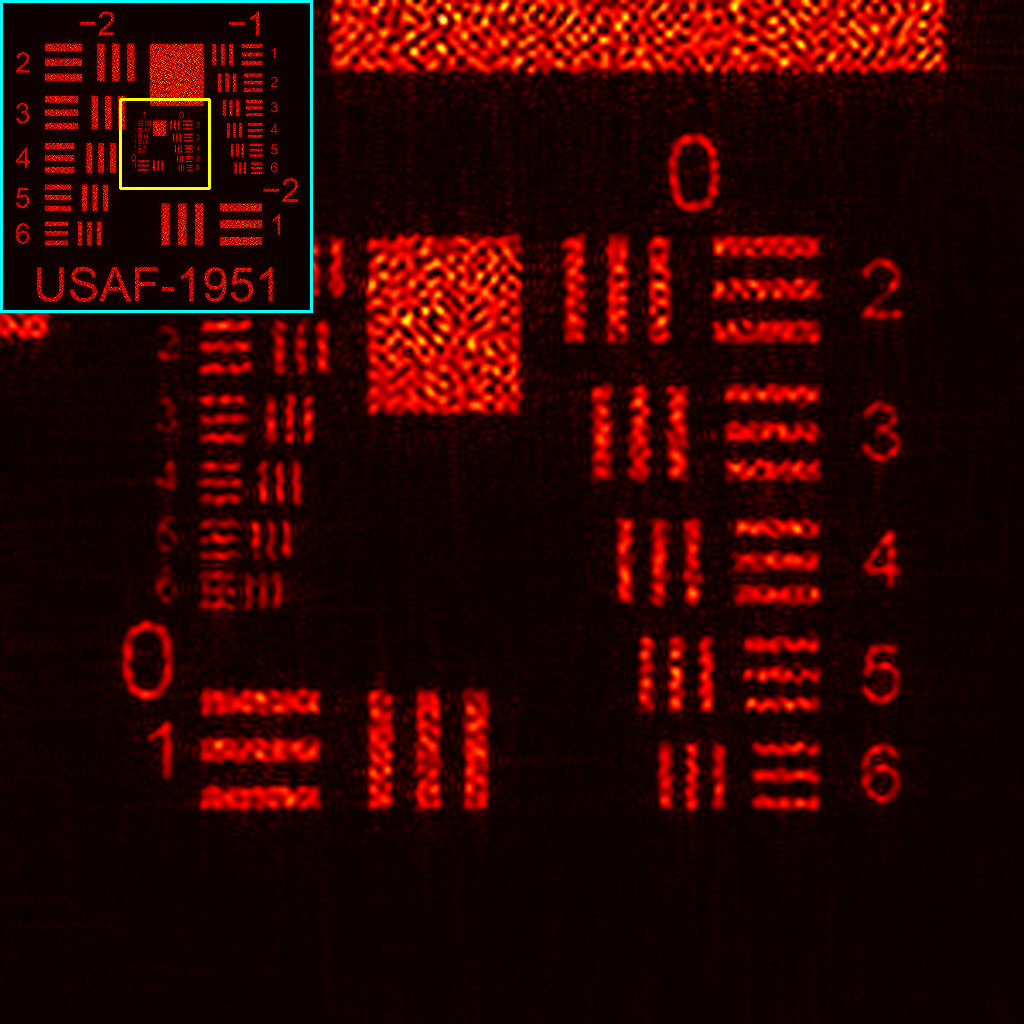}
        \caption{With synchronization.}
    \end{subfigure}
    \caption{\textbf{Effect of synchronization.} (a) and (b) are equalized images without and with synchronization, respectively. For clarity, cropped region from the middle is shown, cropped area is indicated on the inset at top-left. Synchronization improves resolution.}
    \label{fig:effect-of-synchronization}
\end{figure}
\subsection{Coherent Integration Across Apertures}

Assume we have estimates, $\{\hat{\phi}_{\ell,s}\}_{s=1...S,\ell=1...L}$, for each of these phase offsets. 
(Section~\ref{sec:synchronization} describes a procedure for estimating these phase offsets.)
With these offsets in hand, we coherently integrate the subaperture measurements to form a synthetic aperture measurement. 
We then map this measurement to the image plane (perform an inverse Fourier transform) and take the intensity of the result:
\begin{equation}
    \tilde{I}_s = \left|\Fi\left[\sum_{\ell=1}^L U_{\ell,s}\cdot e^{-j\hat{\phi}_{\ell,s}}\right]\right|^2.
    \label{eq:composite-measurement}
\end{equation}

If the phase offsets between subapertures are known precisely, this process produces an image indistinguishable to what one would have captured had you imaged the scene with a large lens with composite aperture $\tilde{A}=\sum_{\ell=1}^L A_\ell$. 
That is
\begin{equation}
    \tilde{I}_s = \left|\Fi\left[\tilde{A} \odot \F(x \odot \eta_s)\right]\right|^2.
    \label{eq:composite-aperture-imaging}
\end{equation}

An image reconstructed in this way using a $3\times3$ field sensor array is shown in Fig.~\ref{fig:intermediate-reconstruction-b}, and the intensity of the field captured by the center field sensor is shown in Fig.~\ref{fig:intermediate-reconstruction-a}. 
The composite aperture clearly provides a resolution improvement. 
Thus, if one can synchronize the apertures, snapshot imaging systems can perform SA imaging in the presence of speckle. 
However, the reconstructed image still contains significant speckle noise. 

\subsection{Incoherent Averaging Across Speckle Realizations}
\label{sec:equalized-speckle-averaging}

To mitigate speckle, we perform image-plane intensity averaging. 
That is, we average the images $\tilde{I}_s$ over speckle realizations. 
When we have a large number of speckle realizations $S$, this averaging procedure provides a useful Monte Carlo approximation of $\tilde{I}_s$'s expectation:
\begin{align}
    \hat{I}=\frac{1}{S}\sum_s \tilde{I}_s \approx \mathbb{E}[\tilde{I}_s].
    \label{eq:summation}
\end{align}

In Appendix~\ref{sec:appendix-b} we prove that $\tilde{I}_s$'s expectation is equal to the scene
albedo $|x|^2$ filtered by the autocorrelation function of the composite aperture $\tilde{A}$. 
That is
\begin{align}
    \mathbb{E}[\tilde{I}_s] = \mathcal{F}^{-1}\left[R_{\tilde{A}} \odot \mathcal{F}(|x|^2)\right]
    \label{eq:corr_response}
\end{align}
where $R_{\tilde{A}}$ is the autocorrelation of the aperture $\tilde{A}$. 
Accordingly, $\hat{I}$ provides us a filtered estimate of $|x|^2$.

Notably, the support of $R_{\tilde{A}}$ (which determines which of the scene's spatial frequencies are recorded) is larger than the support $\tilde{A}$: not only does speckle averaging suppress the speckle noise, it also \emph{broadens the observed area of spectrum while filling in the missing information between subapertures}. 
Consequently, we can perform hallucination-free synthetic aperture imaging {even with non-overlapping apertures}.  

\subsection{Distributed Aperture Phase Synchronization}
\label{sec:synchronization}

In this work we  computationally synchronize the subapertures by maximizing the variance of our (biased) albedo estimate $\hat{I}$:
\begin{align}
    \hat{\Phi}
    &= \argmax_{\Phi}
    \; \Var\bigl(\hat{I}(\Phi)\bigr),
    \label{eq:phase-estimation}
\end{align}
where $\Phi = \{ \phi_{\ell,s} : \ell = 1,\dots,L,\; s = 1,\dots,S \}$ and $\hat{I}$ is defined according to averaging in~\eqref{eq:summation}.
We employ gradient ascent to solve this optimization problem. 
Fig.~\ref{fig:synchronization} depicts a diagram of our computational phase synchronization method. 
The importance of this synchronization is illustrated in Fig.~\ref{fig:effect-of-synchronization}, which illustrated coherently integrated measurements with and without synchronization.  
Note that all other figures in this manuscript use synchronized aperture measurements.

\subsection{Equalization Filtering}

Speckle averaging produces a filtered estimate of the true scene albedo. 
(Filter responses are illustrated in Fig.~\ref{fig:theory-illustration}b and Fig.~\ref{fig:theory-illustration}c.) 
To compensate for these responses, we can (optionally) apply an equalization filter 
\begin{equation}
    H_{eq}(f) = 
        \begin{cases}
        \frac{1}{R_{\tilde{A}}(f)+\sigma}, & \text{if } R_{\tilde{A}}(f) \geq \epsilon \\
        0,  & \text{if } R_{\tilde{A}}(f) < \epsilon,
        \end{cases}
    \label{eq:equalization-filter}
\end{equation}
where $f$ indexes spatial frequency and $\sigma$ and $\epsilon$ are user specified parameters  
used to avoid overly amplifying noise (include finite sample approximation noise) in regions of weak response. 
The estimated albedo of the scene is then given by
\begin{equation}
    \hat{I}_{eq}= \Fi\left[H_{eq} \odot \F(\hat{I})\right].
    \label{eq:white}
\end{equation}

The overall imaging framework is illustrated in Fig.~\ref{fig:imaging-pipeline}. Coherent integration, followed by incoherent intensity averaging, followed by equalization can produce a speckle-free high-resolution image with features well past the diffraction limit of the individual subapertures.

\begin{figure*}[t]
    \centering
    \includegraphics[width=\linewidth]{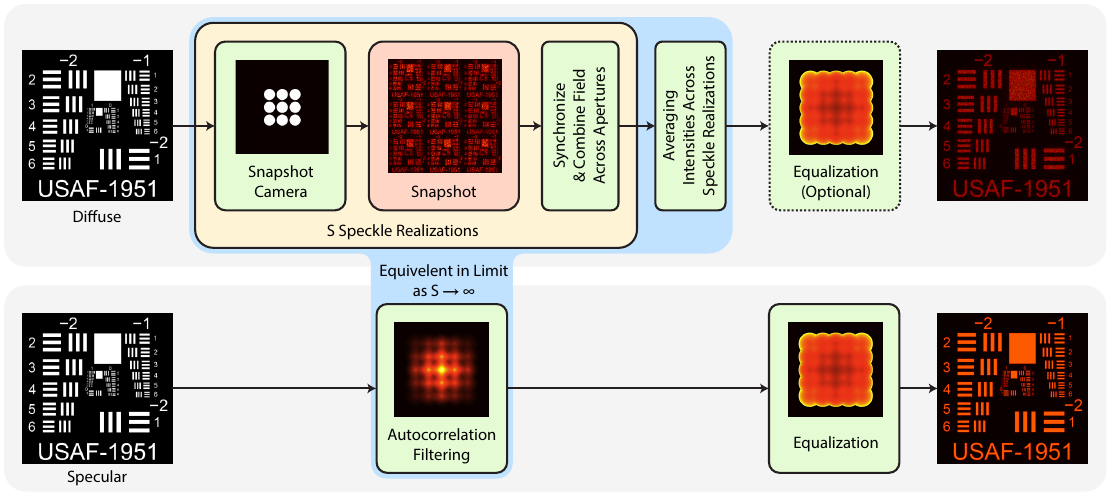}
    \caption{\textbf{Proposed imaging process.} By synchronizing and coherently integrating fields across subapertures and averaging intensities across speckle realizations, our SA imaging framework can perform hallucination-free distributed aperture SA imaging.}
    \label{fig:imaging-pipeline}
\end{figure*}

\newlength{\mylength}   
\setlength{\mylength}{0.198\textwidth}  
\begin{figure*}[t]
    \centering

    \begin{subfigure}[b]{\mylength}
        \includegraphics[width=\linewidth]{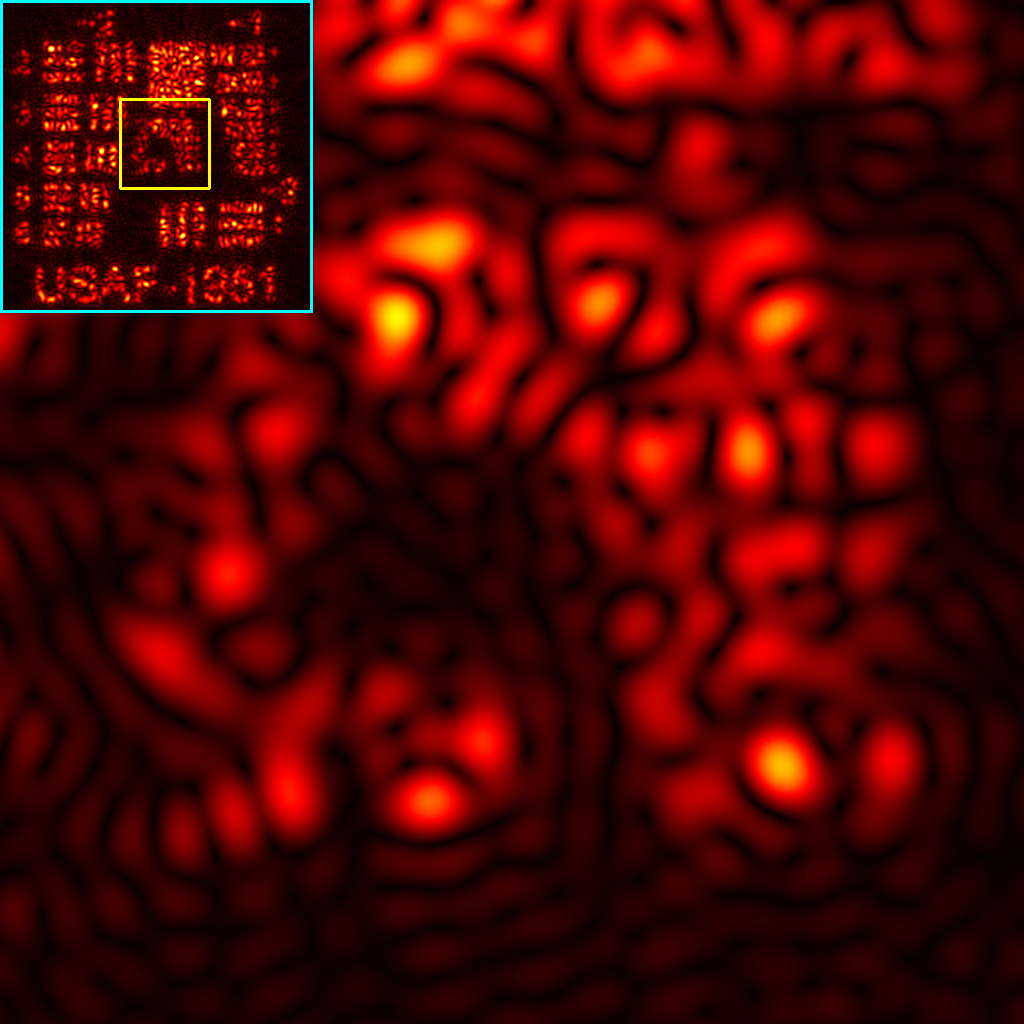}
    \end{subfigure}\hfill
    \begin{subfigure}[b]{\mylength}
        \includegraphics[width=\linewidth]{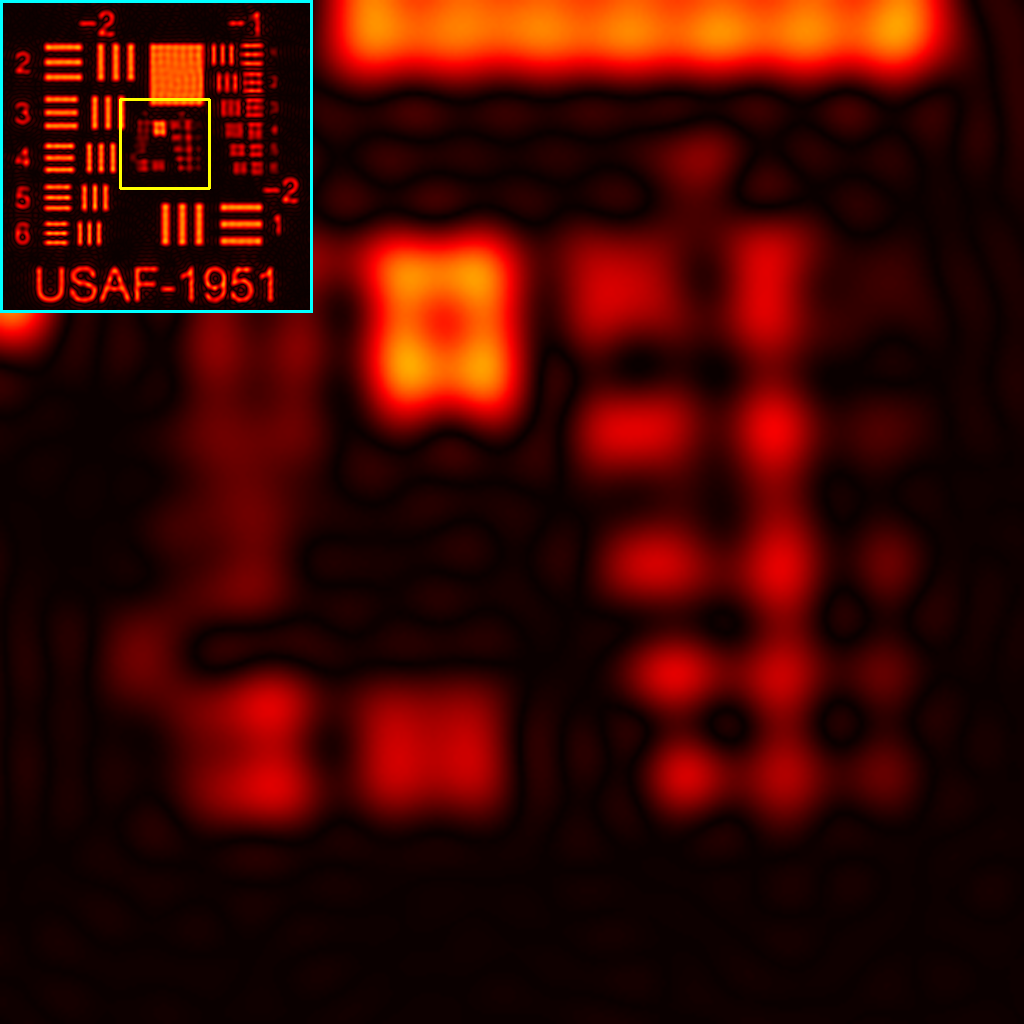}
    \end{subfigure}\hfill
    \begin{subfigure}[b]{\mylength}
        \includegraphics[width=\linewidth]{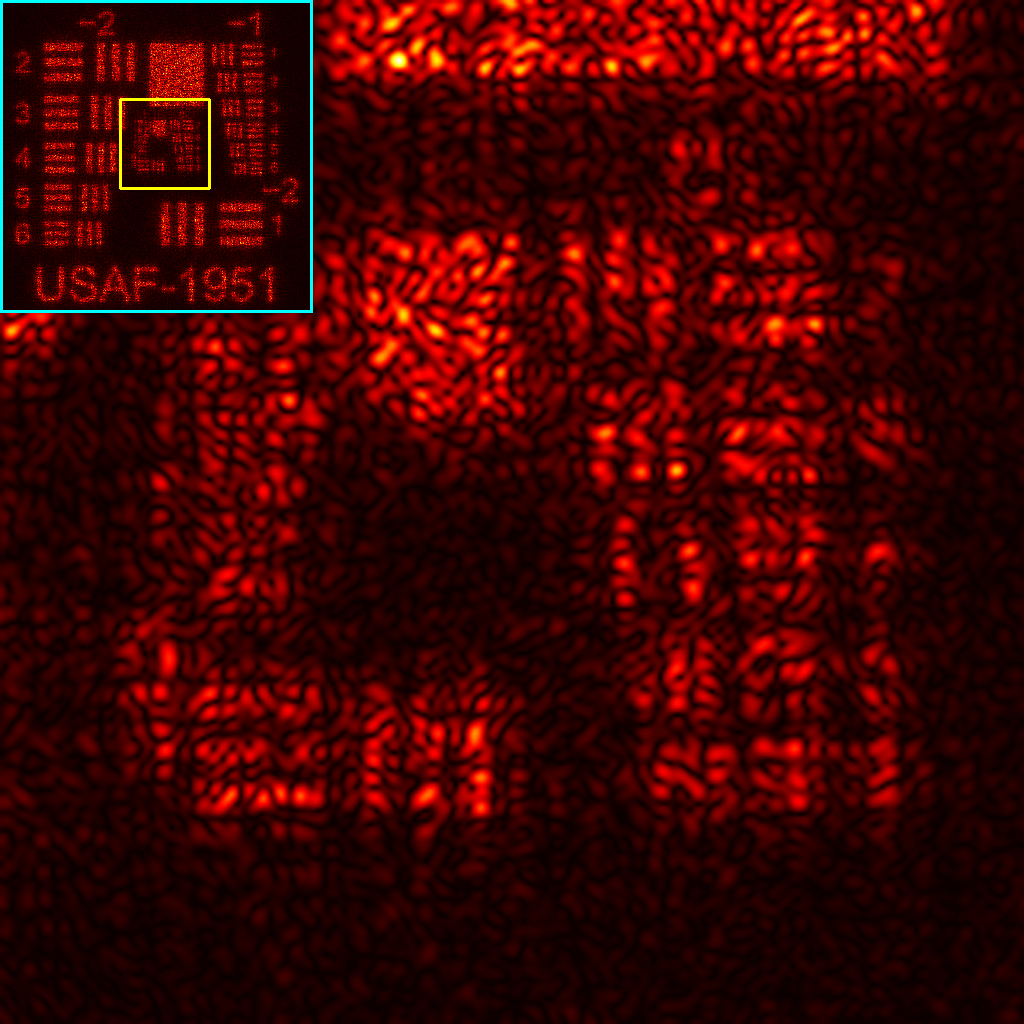}
    \end{subfigure}\hfill
    \begin{subfigure}[b]{\mylength}
        \includegraphics[width=\linewidth]{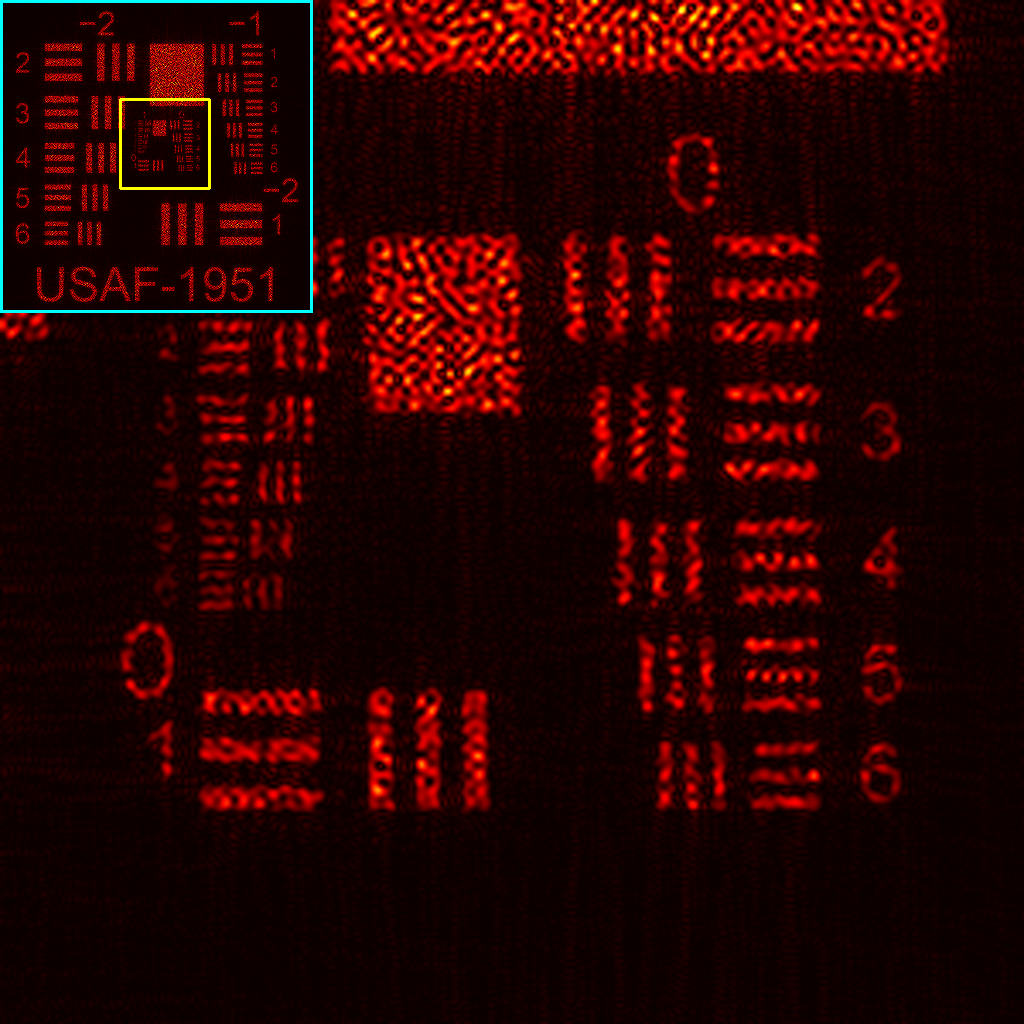}
    \end{subfigure}\hfill
    \begin{subfigure}[b]{\mylength}
        \includegraphics[width=\linewidth]{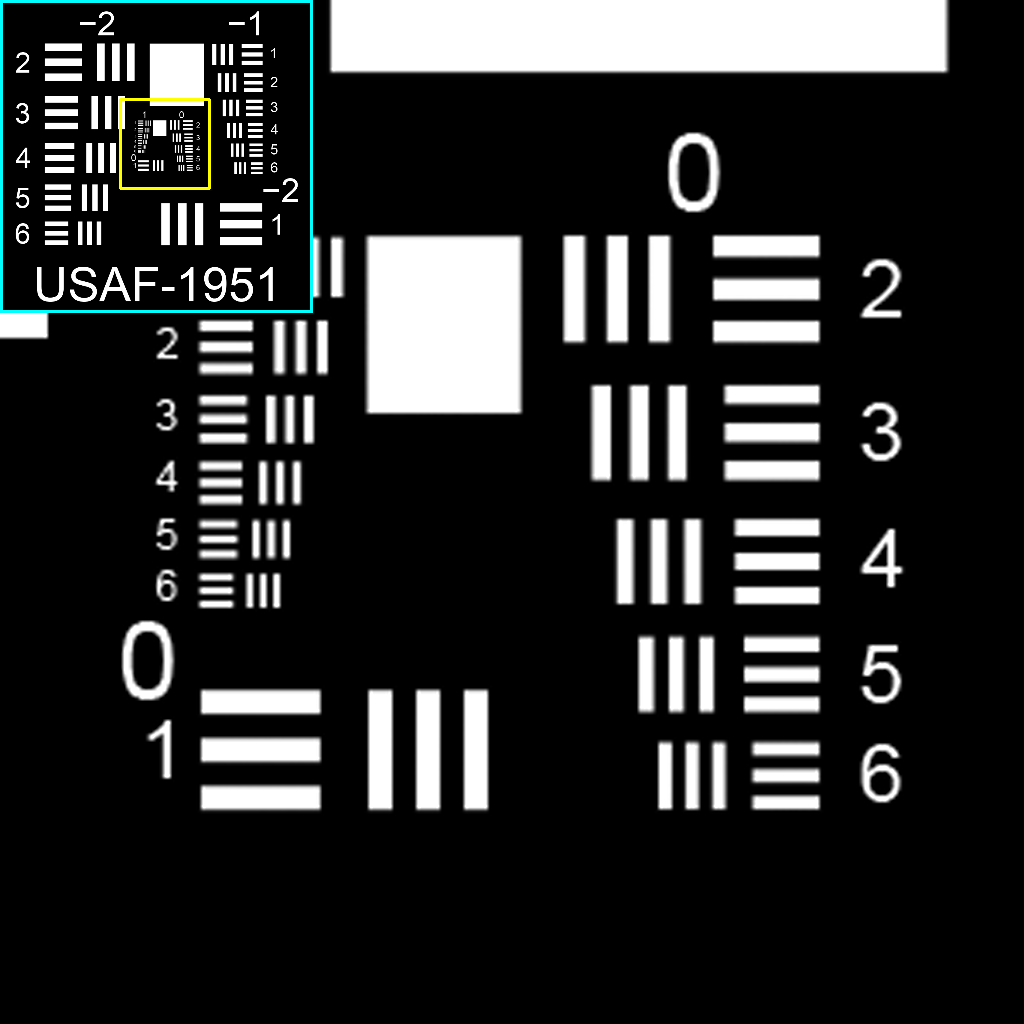}
    \end{subfigure}

    \vspace{0.4mm}

    \begin{subfigure}[b]{\mylength}
        \includegraphics[width=\linewidth]{figures/testudo/center_image_diffuse.png}
    \end{subfigure}\hfill
    \begin{subfigure}[b]{\mylength}
        \includegraphics[width=\linewidth]{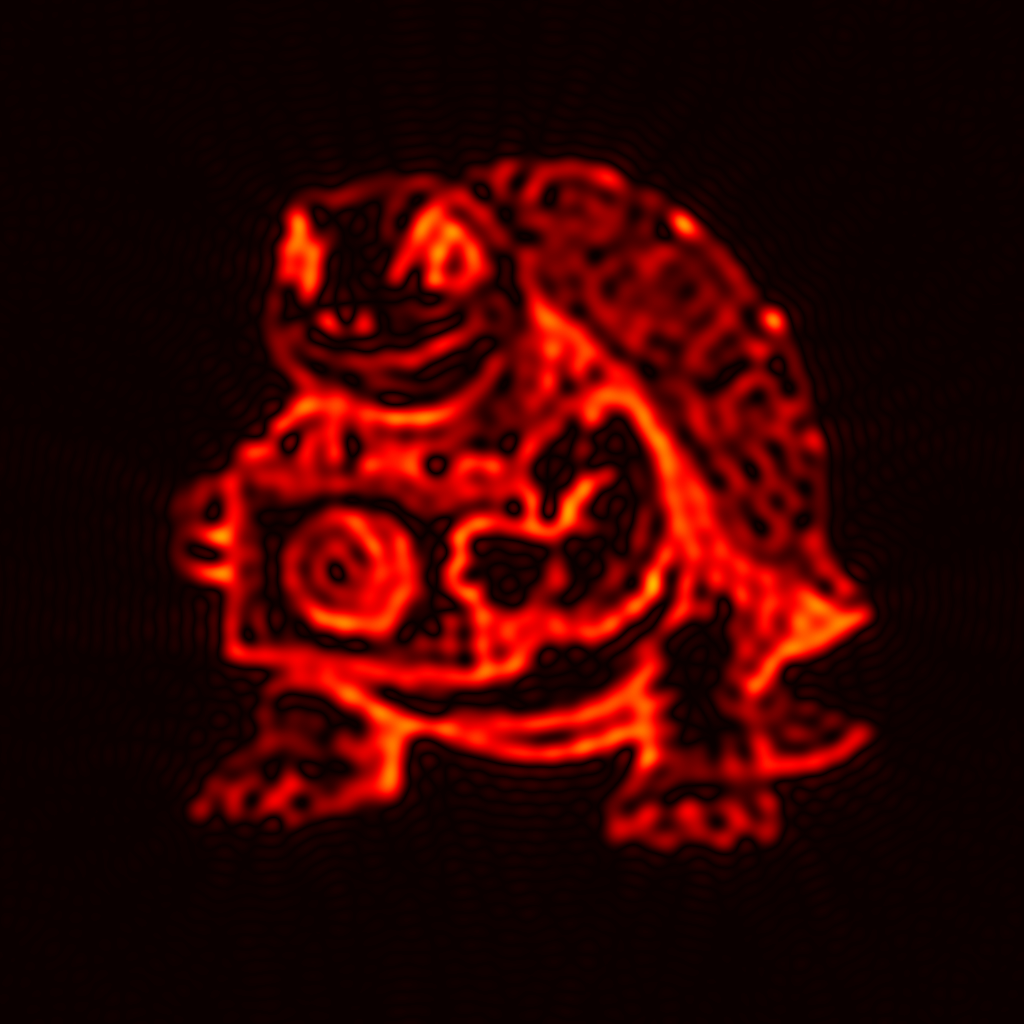}
    \end{subfigure}\hfill
    \begin{subfigure}[b]{\mylength}
        \includegraphics[width=\linewidth]{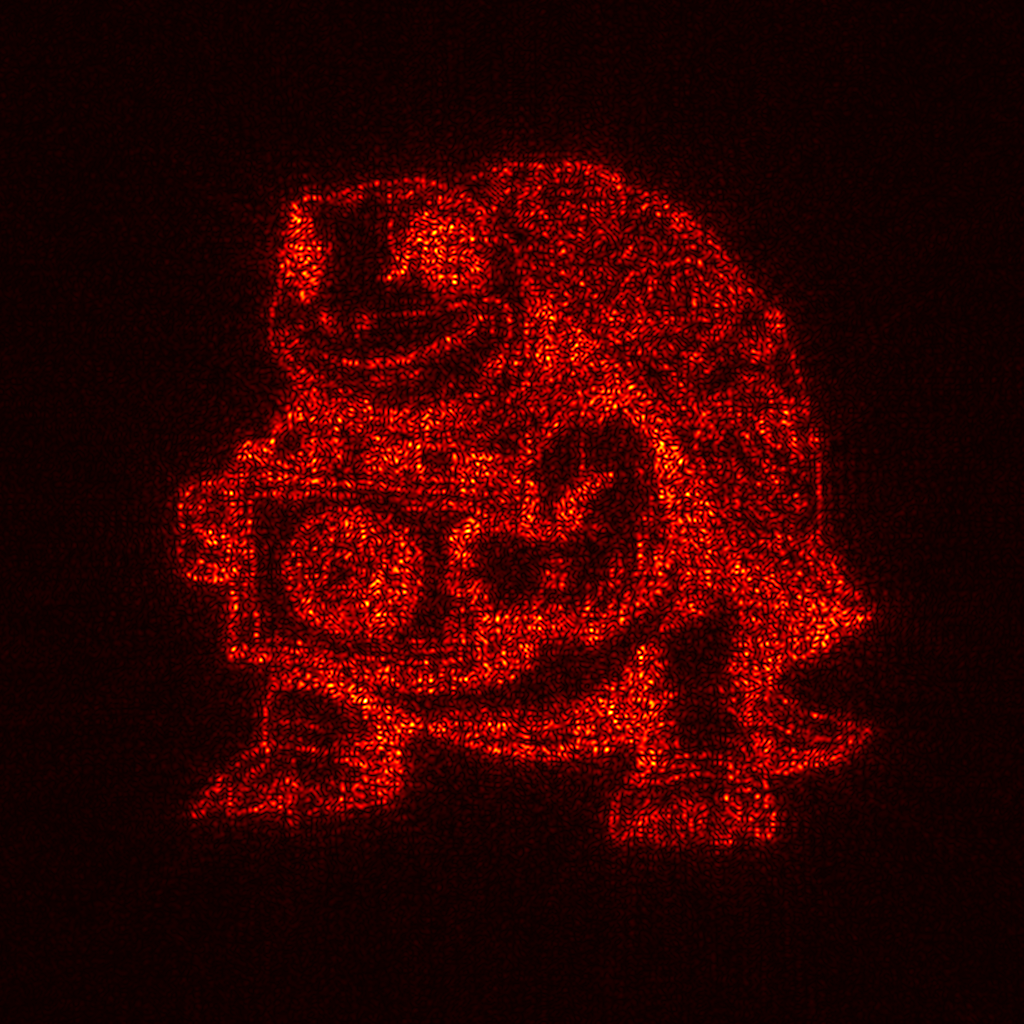}
    \end{subfigure}\hfill
    \begin{subfigure}[b]{\mylength}
        \includegraphics[width=\linewidth]{figures/testudo/equalized_image_r50.png}
    \end{subfigure}\hfill
    \begin{subfigure}[b]{\mylength}
        \includegraphics[width=\linewidth]{figures/testudo/ground_truth_image_gray.png}
    \end{subfigure}

    \vspace{0.4mm}

    \begin{subfigure}[b]{\mylength}
        \includegraphics[width=\linewidth]{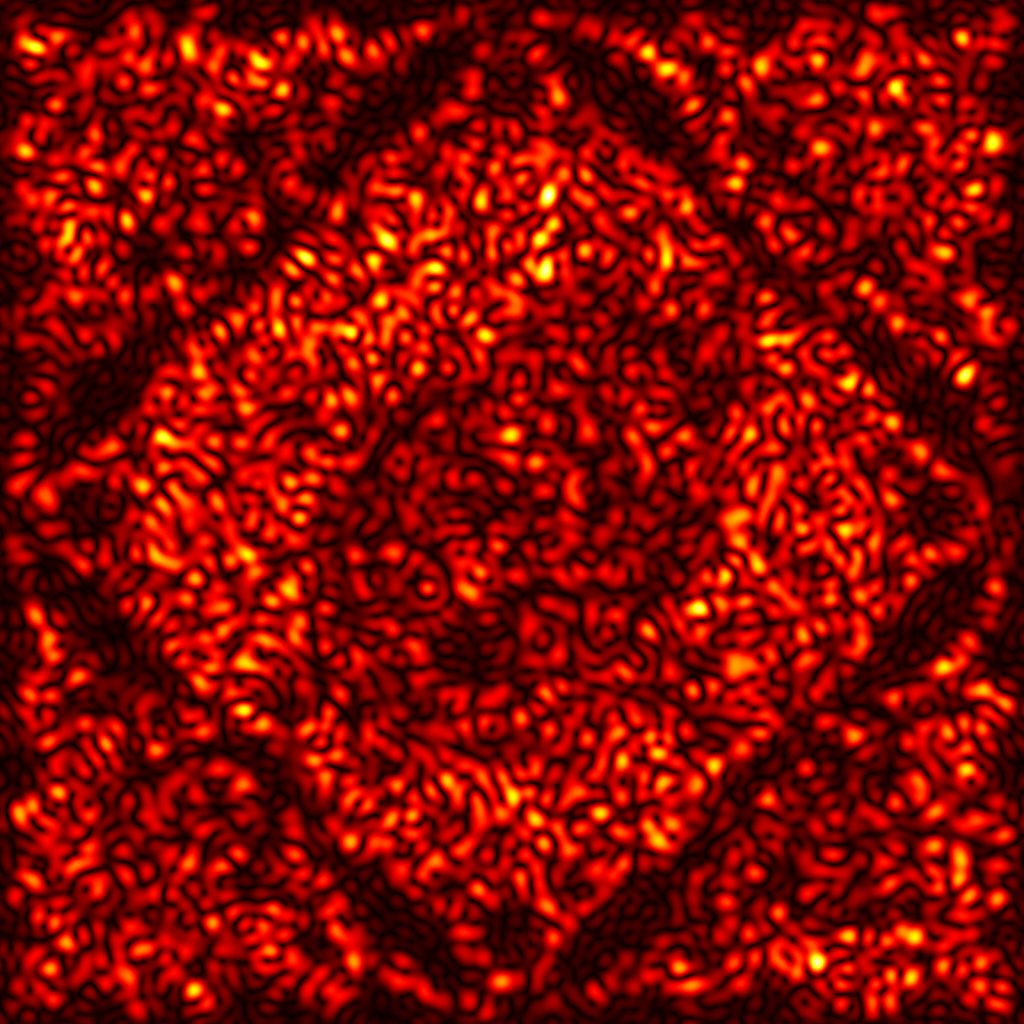}
        \caption{Single-Diffuse}
    \end{subfigure}\hfill
    \begin{subfigure}[b]{\mylength}
        \includegraphics[width=\linewidth]{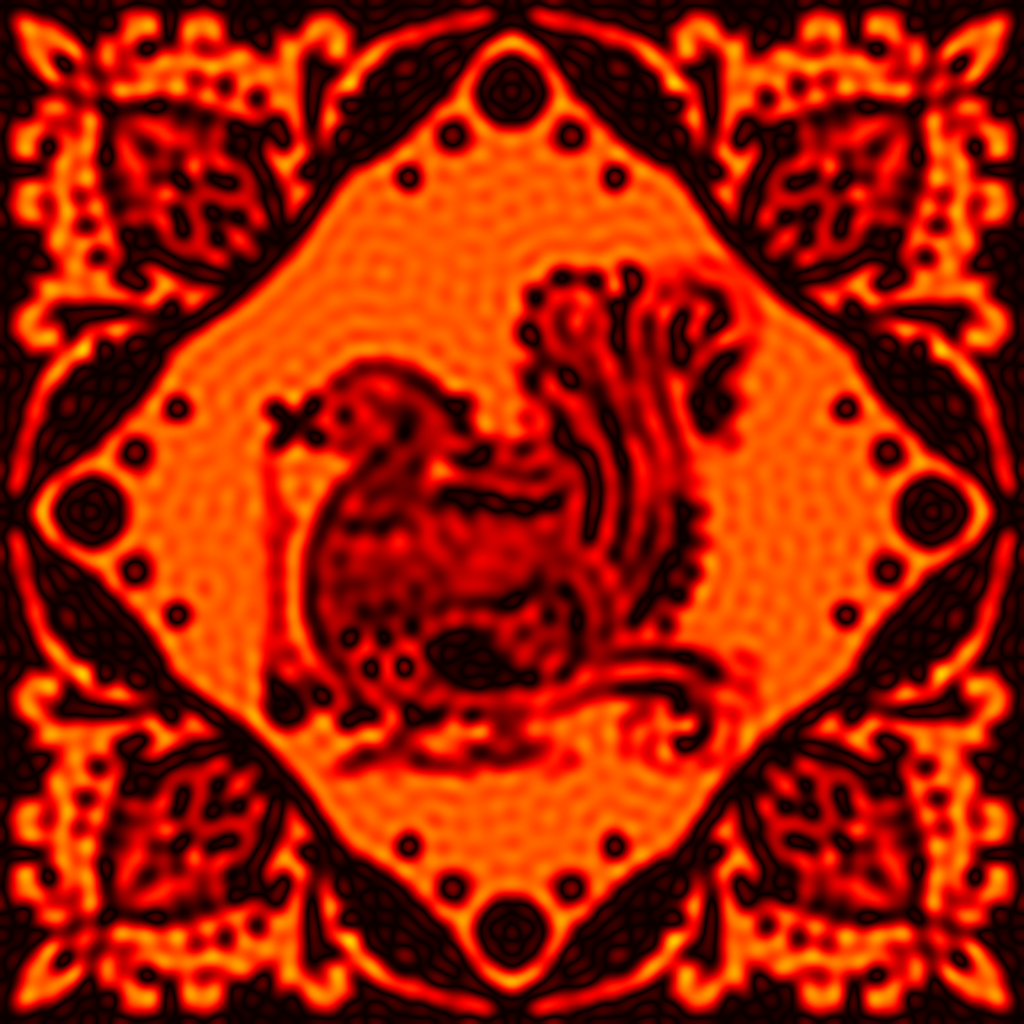}
        \caption{Single-Specular}
    \end{subfigure}\hfill
    \begin{subfigure}[b]{\mylength}
        \includegraphics[width=\linewidth]{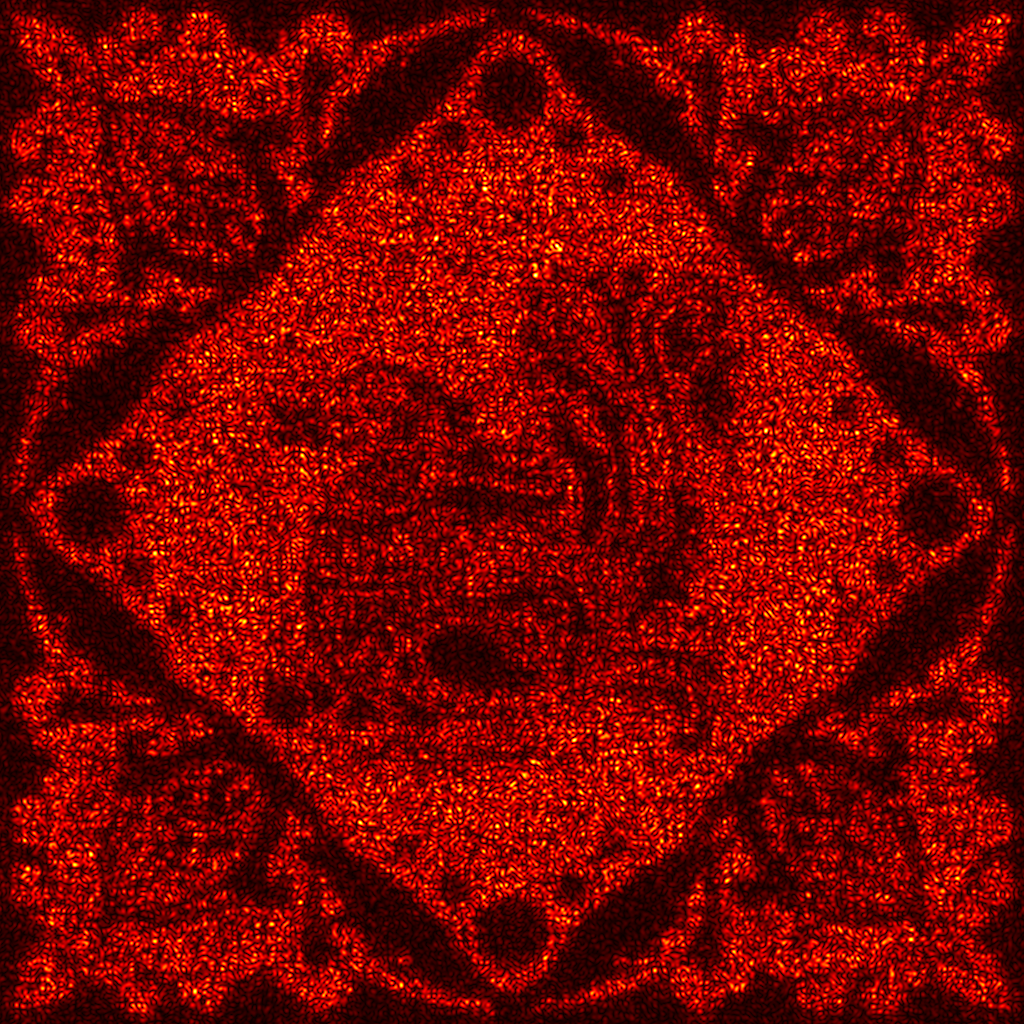}
        \caption{Single-Combined}
    \end{subfigure}\hfill
    \begin{subfigure}[b]{\mylength}
        \includegraphics[width=\linewidth]{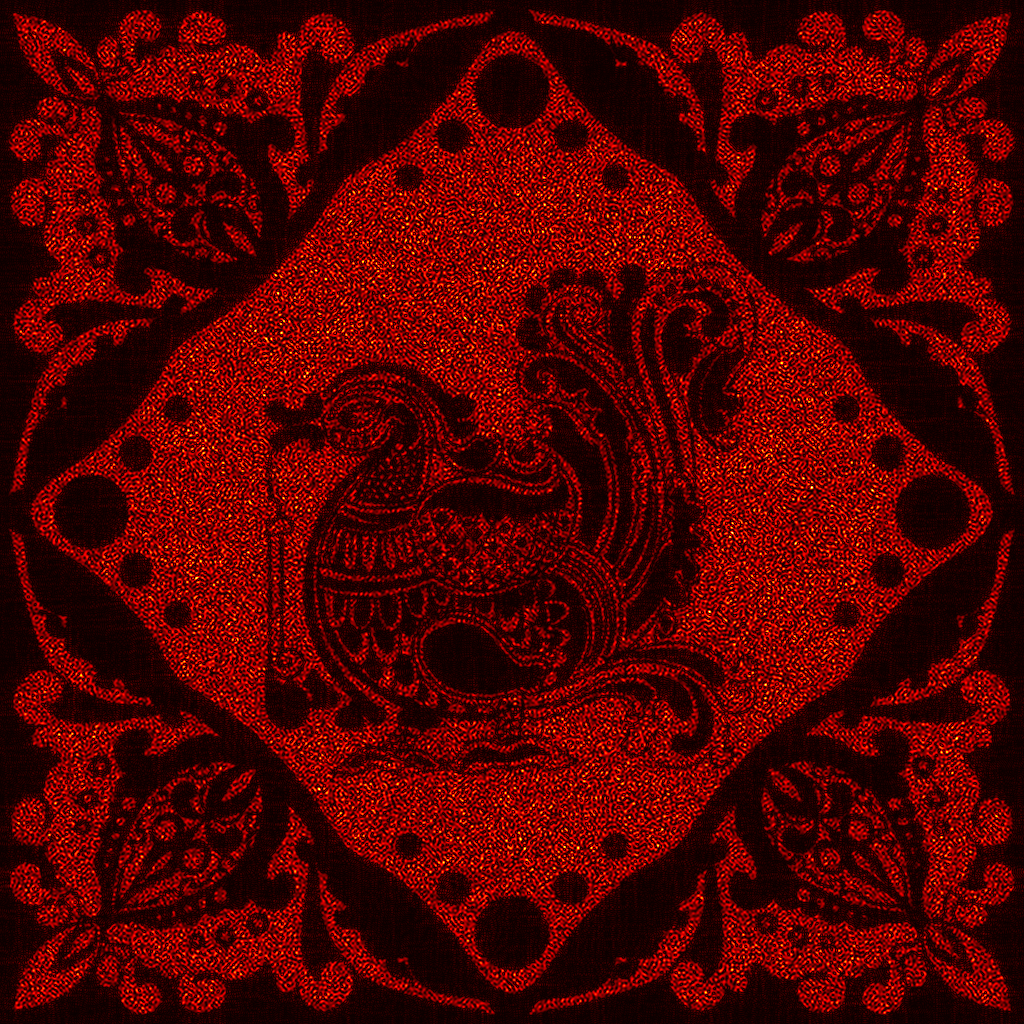}
        \caption{Proposed}
    \end{subfigure}\hfill
    \begin{subfigure}[b]{\mylength}
        \includegraphics[width=\linewidth]{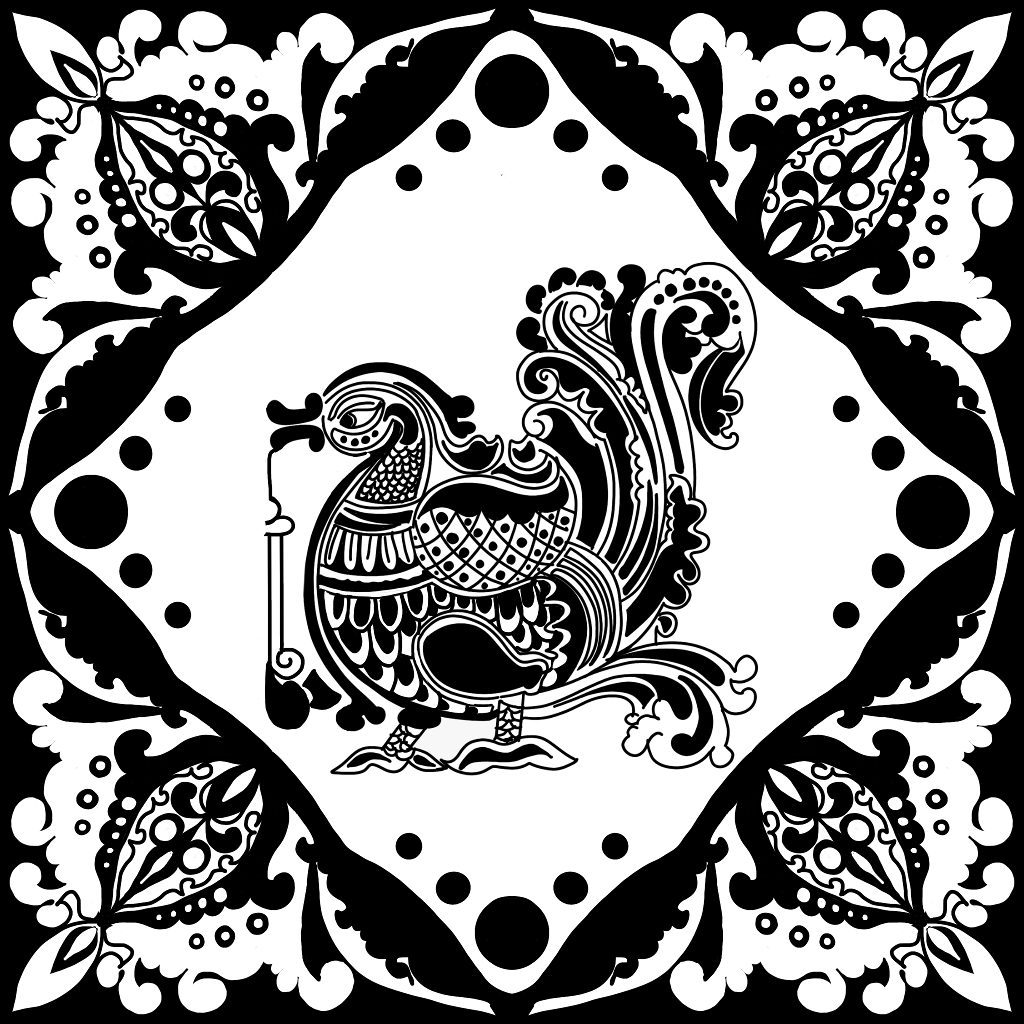}
        \caption{Ground truth}
    \end{subfigure}

    \caption{\textbf{Simulation results.} (a) images of diffuse targets captured from center aperture position; (b) images of the specular targets captured from center aperture position; (c) estimated image by combining multiple aperture images for single speckle realization; (d) estimated image using proposed method 50 speckle realizations; and (e) ground-truth image.
    Our proposed method substantially improves resolution.
    }
    \label{fig:simulations}
\end{figure*}

\begin{figure}[t]
    \centering
    \includegraphics[width=0.48\textwidth]{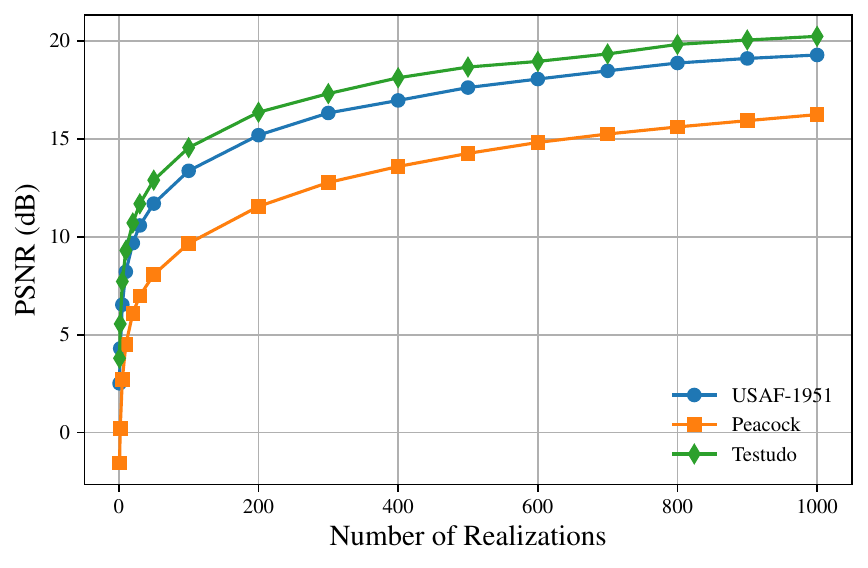}
    \caption{\textbf{PSNR (dB) vs Speckle Realizations}. 
    This shows PSNR (dB) vs number of speckle realizations for the images shown in Fig.~\ref{fig:simulations}. 
    PSNR improves monotonically with the number of speckle realizations.}
    \label{fig:psnr}
\end{figure}

\setlength{\mylength}{0.18\textwidth}  
\begin{figure*}[t]
    \centering

    \begin{subfigure}[b]{\mylength}
        \includegraphics[width=\linewidth]{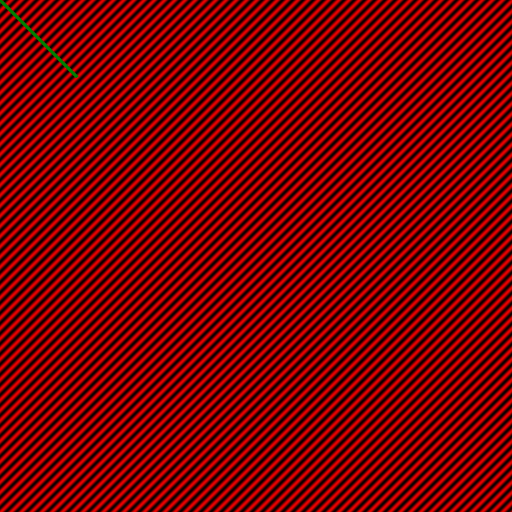}
    \end{subfigure}\hfill
    \begin{subfigure}[b]{\mylength}
        \includegraphics[width=\linewidth]{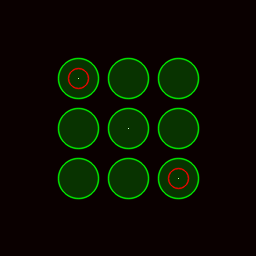}
    \end{subfigure}\hfill
    \begin{subfigure}[b]{\mylength}
        \includegraphics[width=\linewidth]{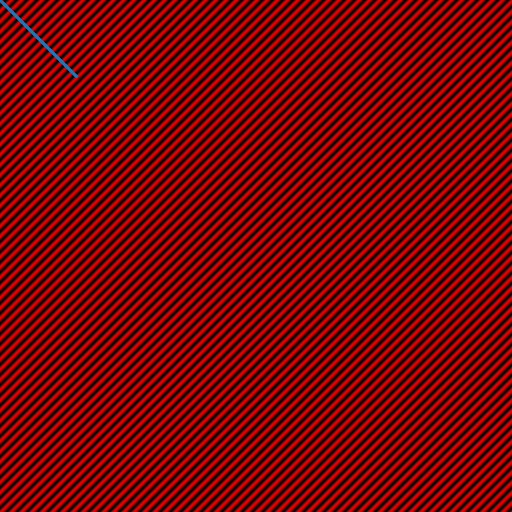}
    \end{subfigure}\hfill
    \begin{subfigure}[b]{\mylength}
        \includegraphics[width=\linewidth]{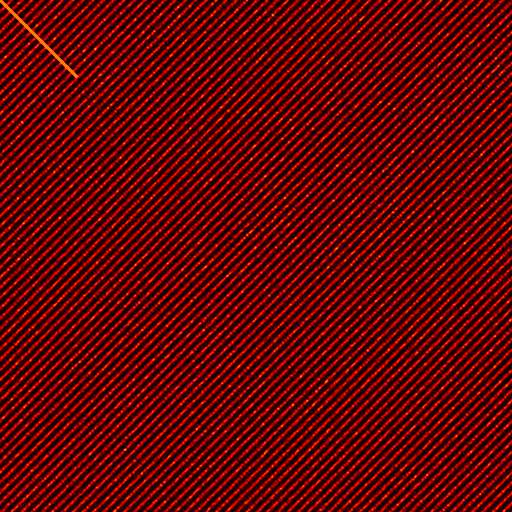}
    \end{subfigure}\hfill
    \begin{subfigure}[b]{0.267\textwidth}
        \includegraphics[width=\linewidth]{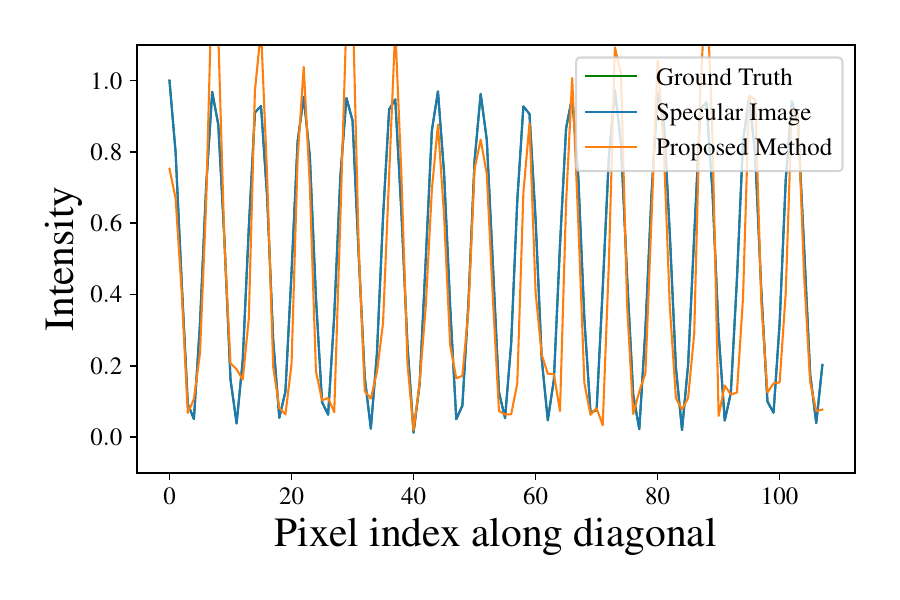}
    \end{subfigure}

    \vspace{0.5mm} 

    \begin{subfigure}[b]{\mylength}
        \includegraphics[width=\linewidth]{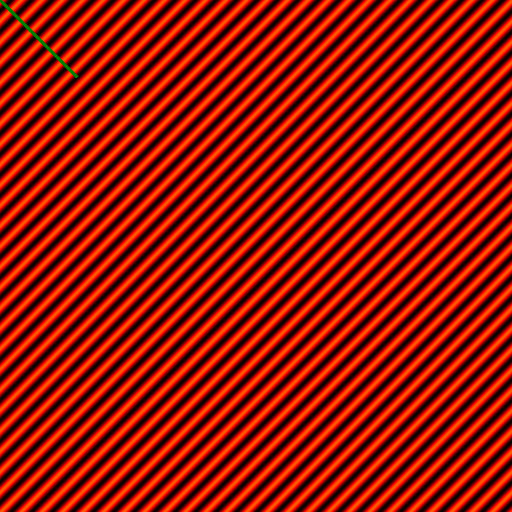}
        \caption{Scene}
    \end{subfigure}\hfill
    \begin{subfigure}[b]{\mylength}
        \includegraphics[width=\linewidth]{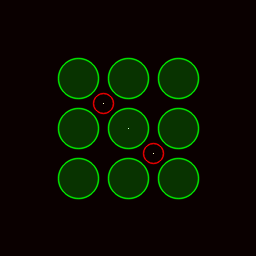}
        \caption{Spectrum}
    \end{subfigure}\hfill
    \begin{subfigure}[b]{\mylength}
        \includegraphics[width=\linewidth]{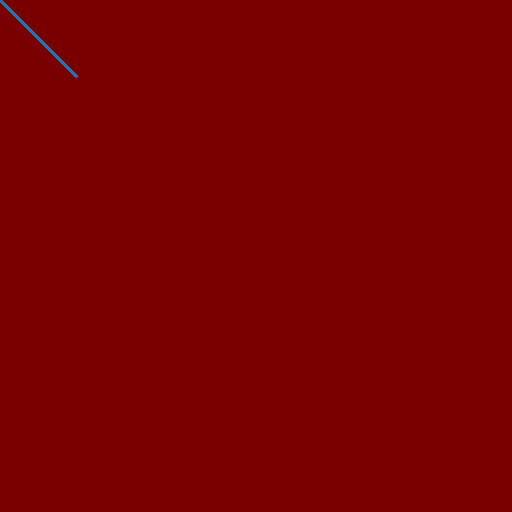}
        \caption{Full Specular}
    \end{subfigure}\hfill
    \begin{subfigure}[b]{\mylength}
        \includegraphics[width=\linewidth]{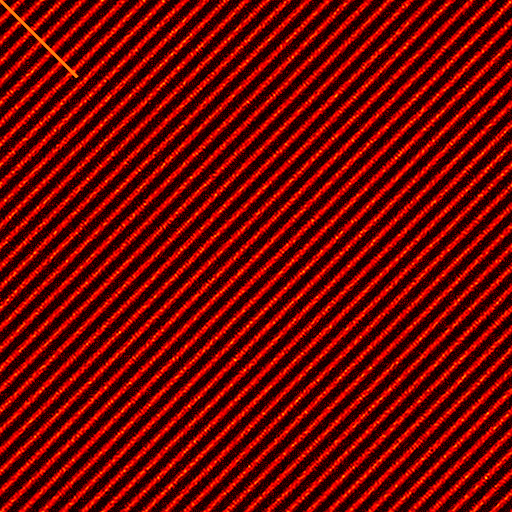}
        \caption{Proposed}
    \end{subfigure}\hfill
    \begin{subfigure}[b]{0.267\textwidth}
        \includegraphics[width=\linewidth]{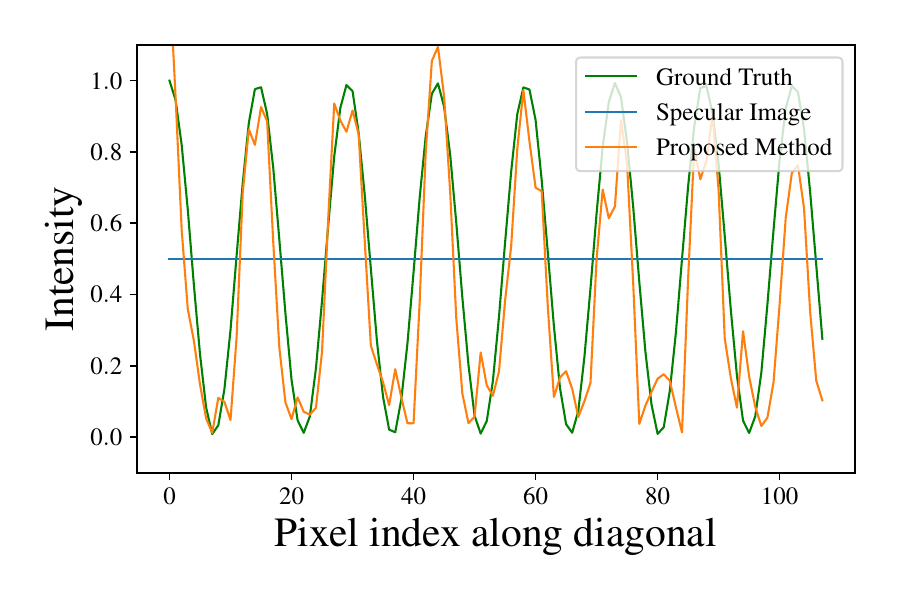}
        \caption{Line Profile}
    \end{subfigure}
    
    \caption{\textbf{Recovering spatial frequency information between subapertures by exploiting speckle}. 
    The first row corresponds to the case where the sinusoid frequency is inside the top-left subaperture and the second row corresponds to the case where the sinusoid frequency is between subapertures. 
    Column (a) illustrates the target image;
    column (b) illustrates the spectrum of the corresponding scene with subapertures superimposed in green and impulses corresponding to sinusoid highlighted in red circles, the image is zoomed in and cropped for better visualization;
    column (c) illustrates an image reconstructed by measuring the field at each subaperture and coherently combining them for a specular scene; 
    column (d) illustrates an image reconstructed using the proposed method with 500 speckle realizations for a diffuse scene;
    column (e) illustrates a segment of a line profile along the main diagonal for both specular image and image reconstructed from the proposed method with a diffuse scene. 
    Frequencies outside the array's passband are completely lost for specular scenes. By contrast, the effective aperture expansion provided by speckle averaging allows our method to recover frequency content outside the passband, without hallucination.
    }
    \label{fig:sinosoids}
\end{figure*}

\section{Simulations and Results}
\label{sec:simulations}

Fig.~\ref{fig:simulations} shows simulation results for three different scenes.
The images for diffuse scenes suffer from speckle noise, whereas the images of specular scenes do not have speckle but instead exhibit lower resolution.
Combining information from multiple subapertures increases spatial resolution; however, it does not effectively eliminate speckle noise inherent in the captured images. 
The proposed method is able to accurately recover high-resolution images of the scene even in the presence of boiling speckle.

Fig.~\ref{fig:psnr} shows the PSNR vs number of speckle realizations for the three scenes shown in~Fig.~\ref{fig:simulations}. 
The PSNR value increases as the number of speckle realizations increases.

Fig.~\ref{fig:sinosoids} demonstrates the capability of the proposed method to recover missing information between subapertures.
Two scenarios were considered: one with sinusoid frequency centered on a subaperture (top-left/bottom-right pair) and another positioned between subapertures, with corresponding scene images shown in the first column.
In both scenarios, the optical field at each subaperture was measured, coherently combined as in~\eqref{eq:composite-measurement}, and converted to intensity images. 
The baseline method assumed a specular scene and the proposed method assumed a diffuse scene, producing speckle that was mitigated by averaging multiple realizations. 
The resulting images from the baseline and proposed methods are shown in the third and fourth columns, respectively, with diagonal line profiles provided in the fifth column.

Both methods are able to capture frequency content when the frequency falls within the subapertures, as evident from the first row.
However, when the desired frequency component lies between the subapertures, the first method with a specular scene fails to capture the desired frequency content (sinusoid).
By contrast, the proposed method successfully recovers the desired frequency component (sinusoid) despite the frequency being located between subapertures. 

This example illustrates how the proposed method enables the capture of frequency components situated between subapertures---which would otherwise be missed by a sensor array with non-overlapping apertures---by leveraging speckle noise to our advantage.

\section{Discussion and Limitations}\label{sec:discussion}

The averaging of combined field intensities over different speckle realizations is the same as incoherent imaging if one has a single camera with the composite subaperture array as in ours. 
Particularly, incoherent imaging is equivalent to the application of Optical Transfer Function (OTF), which is the autocorrelation function of the aperture \cite{goodman_introduction_2017}.
For a single connected aperture, it is well known that incoherent imaging can produce a two fold resolution improvement over coherent imaging~\cite{goodman_introduction_2017}.

However, if one performs incoherent imaging with an independent array of cameras, they will produce images filtered by OTF and will have the same spectral information in all the images, as we saw in~\Cref{sec:impossibility}.
Hence, it will be impossible to get the resolution improvement that our imaging system produces because, no matter how you combine those individual images, they all capture the same low frequency spectral information.
Therefore, coherent combination of the individual optical fields at each subaperture and then doing the intensity averaging is key to our imaging system.

To our knowledge, this is the first work that demonstrates one can recover spectral information outside the passband of an SA imaging system---without hallucinations. 

Our method does have a major limitation however: It requires measuring the complex-valued optical field at each of the individual subapertures. 
While this can be done using holographic or coded-aperture--based systems~\cite{wu_wish_2019}, each introduces substantial system cost and complexity. 
Extending our method to intensity-only measurements, e.g.,~Fourier ptychography, is non-trivial.

\section{Conclusion}\label{sec:conclusion}

In this work, we have introduced new theory and algorithms to characterize and extend the fundamental limits of light-based synthetic aperture imaging. 
Specifically, we first demonstrated that the distributions of far field measurements are independent of the aperture position under per-measurement-independent speckle, and thus synthetic aperture imaging for fully diffuse targets and sequential capture with per-measurement-independent speckle is impossible. 
We then introduced a snapshot imaging framework that overcomes these limitations. 
Our framework uses a variance maximization procedure to synchronize and coherently integrate fields across subapertures. 
It then performs incoherent averaging of image-plane intensities to mitigate and exploit speckle---our system can exploit speckle to recover frequency content outside the system's pass-band.
While our present results are simulation only, we believe our framework represents an important step towards enabling macroscale light-based synthetic aperture imaging. 


%


\appendices
\renewcommand{\theequation}{\thesection.\arabic{equation}} 
\setcounter{equation}{0} 
\section{Proof of the Translation Invariance}
\label{sec:appendix-a}
In this section, we prove that the distribution of far-field measurements are invariant to the aperture translation under per-measurement-independent speckle for fully diffuse scenes.
The proof uses the convolution property of Fourier transformation and the rotational invariance of circular Gaussian distribution, which maintains that multiplying by a phase factor does not change the statistics of the circular Gaussian.\\

According to Fraunhofer diffraction theory, far-field measurements of a complex scene $x$ that contaminated by speckle noise $\eta$ captured using an aperture $A$ positioned at $\ell$\textsuperscript{th} position can be modeled as
\begin{equation}
    U_{\ell} =A_\ell \odot \F(x \odot \eta),
    \label{eq:forward-model}
\end{equation}
where $\F$ is the Fourier transformation operator and $\odot$ denotes the Hadamard product.

Using the convolution property of the Fourier transform, the inverse Fourier transform of $U_{\ell}$ can be written as 
\begin{equation}
    u_{\ell} = a_\ell \ast (x\odot \eta),
\end{equation}
where $*$ is the convolution operation. 
Since the Fourier transformation is invertible for the cases we are interested in, showing $u_{\ell}$ is translation invariant will imply $U_{\ell}$ is also translation invariant. We consider $u_{\ell}$ for brevity.
\\\\
For any fixed aperture shape $A$, the dependence of the translation $\ell$ in $A_\ell$ can be modeled as
\begin{equation}
    A_\ell[n] = A[n] \ast \delta[n - \ell],
\end{equation}
where $A$ is the same aperture located at the center. 
\\
By taking the Fourier transform of that, we can write
\begin{equation}
    a_\ell[k] = a[k] \cdot e^{j\frac{2\pi k\ell}{N}},
\end{equation}
where $a$ is the Fourier transform of $A$ and  $j$ is $\sqrt{-1}$.
\\\\
Next consider the $n$\textsuperscript{th} element of $u_{\ell}$
\begin{align}
    u_{\ell}[n]
    & = \sum_k a_\ell[k]\cdot x[n-k] \cdot \eta[n-k]\notag \\
    & = \sum_k a[k]\cdot e^{j\frac{2\pi k\ell}{N}}\cdot x[n-k] \cdot \eta[n-k]\notag \\
    & = \sum_k a[k]\cdot x[n-k] \cdot \eta[n-k] \cdot e^{j\frac{2\pi k\ell}{N}}
\end{align}
Let us define another measurement as corresponding to another position $\ell'$ as 
\begin{equation}
    u_{\ell'}[n] =  \sum_k a[k]\cdot x[n-k] \cdot \eta'[n-k] \cdot e^{j\frac{2\pi k\ell'}{N}}
\end{equation}
where, $\eta'[n-k] \sim CN(0,1).$

We know that $\eta \sim CN(0,I)$, which implies, $\eta[n-k]$ is a sample from a circularly symmetric complex gaussian, i.e., $\eta[n-k] \sim CN(0,1)$. Circularly symmetric complex Gaussian is rotation invariant. i.e., its distribution is invariant to the multiplications of unit magnitude complex numbers. As a consequence 
\begin{equation}
    \eta[n-k] \cdot e^{j\frac{2\pi k\ell}{N}} \sim CN(0,1),
\end{equation}
and
\begin{equation}
    \eta'[n-k] \cdot e^{j\frac{2\pi k\ell'}{N}} \sim CN(0,1).
\end{equation}
This means 
\begin{equation}
    p(\eta[n-k]\cdot e^{j\frac{2\pi k\ell}{N}}) = p(\eta'[n-k]\cdot e^{j\frac{2\pi k\ell'}{N}}).
\end{equation}
Therefore,
\begin{equation}
    p(u_{\ell}|A_\ell) = p(u_{\ell'}|A_{\ell'}) \quad\forall \ell,\ell'.
\end{equation}
From this it follows that
\begin{equation}
    p(U_{\ell}|A_\ell) = p(U_{\ell'}|A_{\ell'}) \quad\forall \ell,\ell'.
\end{equation}
That is, $p(U_{\ell}|A_\ell)$ is independent of the aperture location---translation invariant. 
The extension of this result to two-dimensions is straightforward.

\section{Proof of the Aperture Expansion by Speckle Averaging}
\label{sec:appendix-b}

In this section, we prove that the averaging the intensity images across speckle realizations and applying the inverse filter is equivalent to applying a band-pass filter with the expanded support of the autocorrelation of the aperture of the imaging system to the magnitude square of the specular image. 
\\
\newline
Let us analyze the average images obtained as
\begin{equation}
    \hat{I_\ell} = \E[I_\ell] = \E\left[|\Fi\left[A_\ell \odot \F(x \odot \eta)\right]|^2 \right],
\end{equation}
where the expectation is over speckle realizations $\eta$.

According to our previous proof, this expectation is independent of the aperture position $\ell$. Hence, it can be further simplified as 
\begin{equation}
    \hat{I} = \E[I_\ell] = \E\left[|\Fi\left[A \odot \F(x \odot \eta)\right]|^2 \right],
\end{equation}
for any arbitrary aperture position $\ell$.
This is equivalent to
\begin{equation}
    \hat{I} = \E\left[ |a \ast (x\odot \eta) |^2\right].
\end{equation}
Let us simplify this further by considering the $n$\textsuperscript{th} element of $\hat{I}$.

\begin{align}
\hat{I}[n]
    & = \E\left[ | a \ast (x\odot \eta)[n] |^2 \right] \notag\\
    & = \E\left[  \left(a \ast (x\odot \eta)[n]\right) \cdot\left(a \ast (x\odot \eta)[n]\right)^\ast   \right] \notag\\
    & = \E\Bigg[ \sum_{k}^{N-1}\sum_{k'}^{N-1} a[n-k]\cdot a^\ast[n-k'] \notag \\
    &\hspace{10em}\cdot x[k] \cdot \eta[k] \cdot x^\ast[k'] \cdot \eta^\ast[k']\Bigg] \notag\\
    & = \sum_{k}^{N-1}\sum_{k'}^{N-1} a[n-k]\cdot a^\ast[n-k'] \notag\\
    &\hspace{10em}\cdot x[k] \cdot x^\ast[k']  \cdot \E\left[\eta[k] \cdot \eta^\ast[k']\right] \notag\\
    & = \sum_{k}^{N-1}\sum_{k'}^{N-1} a[n-k]\cdot a^\ast[n-k'] \cdot x[k] \cdot x^\ast[k']  \cdot R_\eta[k-k'] \notag\\
    & = \sum_{k}^{N-1}\sum_{k'}^{N-1} a[n-k]\cdot a^\ast[n-k'] \cdot x[k] \cdot x^\ast[k']  \cdot \delta[k-k'] \notag\\
    & = \sum_{k}^{N-1} a[n-k]\cdot a^\ast[n-k] \cdot x[k] \cdot x^\ast[k] \notag\\
    & = \sum_{k}^{N-1} |a[n-k]|^2 \cdot |x[k]|^2 \notag\\
    & = \left(|a|^2 \ast |x|^2\right)[n]
\end{align}
Where,
\begin{equation}
    \E\left[\eta[k] \cdot \eta^\ast[k']\right] = R_{\eta}[k-k'] = \delta [k-k']
\end{equation}
Therefore, 
\begin{equation}
    \hat{I} = |a|^2 \ast |x|^2.
\end{equation}
This is equivalent to
\begin{equation}
    \hat{I} = \Fi\left[R_A \odot \F(|x|^2)\right].
\end{equation}
where $R_A$ is the auto-correlation of the aperture $A$, and it has 2x bandwidth as $A$. 

We should do inverse filtering (or equalization) to correct distortions caused by the auto-correlation filter.
\begin{equation}
    \hat{I}_{eq}= \Fi\left[R_A^{-1} \odot \F(\hat{I})\right].
\end{equation}
Then the overall effect will be
\begin{equation}
    \hat{I}_{eq} = \Fi\left[A_{eff} \odot \F(|x|^2)\right],
\end{equation}
where $A_{eff}$ (effective transfer function) is the aperture  with the expanded support of the autocorrelation of the aperture of the imaging system.

It is important to note that the imaging aperture can have any arbitrary shape.
This proof also shows that, by averaging, we can achieve two-fold resolution improvement in the presence of fully developed speckle.

\section*{Acknowledgment}

We thank Chathurya Bombuwala for creating the illustrations used in this work. 



\newpage

%

\begin{IEEEbiography}[{\includegraphics[width=1in,height=1.25in,clip,keepaspectratio]{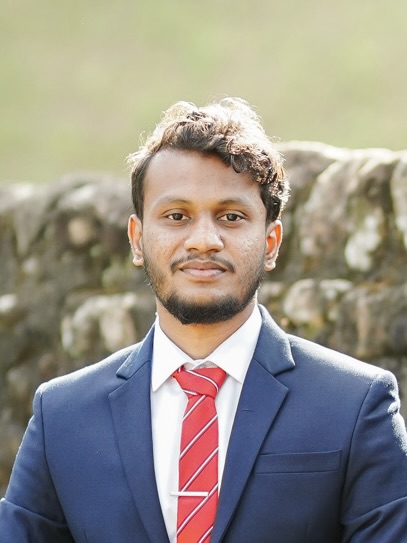}}]{Janith B. Senanayaka} is a PhD student in Electrical and Computer Engineering at the University of Maryland, College Park, USA. He was awarded the Dean's Fellowship in the first year of his studies. He received his BSc in Electrical and Electronic Engineering with First Class Honors from the University of Peradeniya, Sri Lanka. His research interests are in computational imaging, signal processing, and machine learning. He places special emphasis on solving inverse problems and developing robust reconstruction algorithms for imaging systems. Outside of his academic pursuits, he also enjoys traveling, hiking,  reading, painting, photography, music and volunteering.
\end{IEEEbiography}
\begin{IEEEbiography}[{\includegraphics[width=1in,height=1.25in,clip,keepaspectratio]{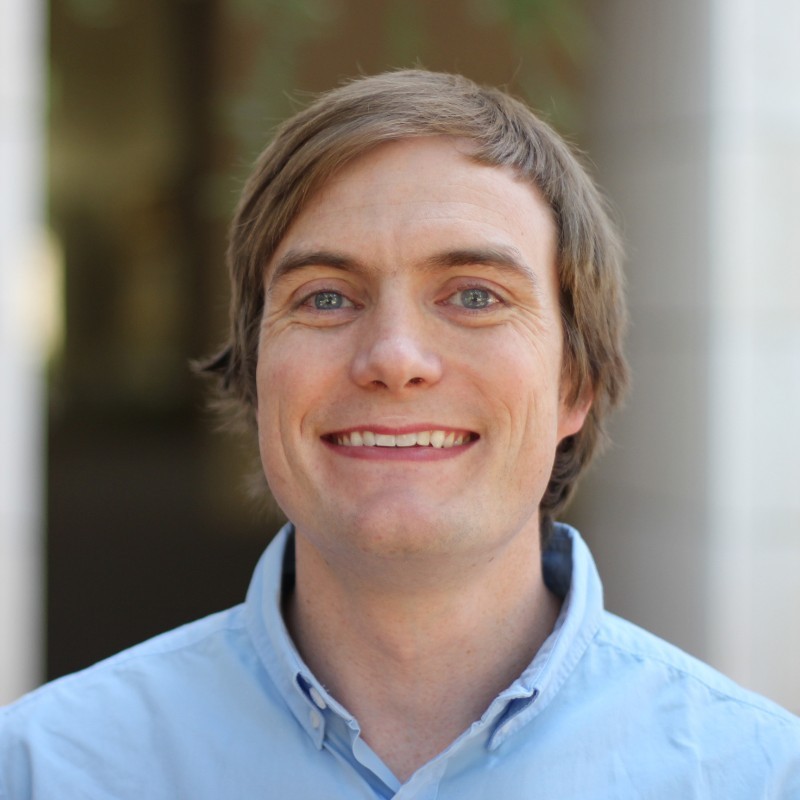}}]{Christopher A. Metzler} is an Assistant Professor in the Department of Computer Science at
the University of Maryland College Park, where he leads the UMD Intelligent Sensing Laboratory. He is a member of the University of Maryland Institute for Advanced Computer Studies (UMIACS) and has a courtesy appointment in the Electrical and Computer Engineering Department. His research develops new systems and algorithms for solving problems in computational imaging and sensing, machine learning, and wireless communications. His work has received multiple best paper awards; he recently received NSF CAREER, AFOSR Young Investigator Program, and ARO Early Career Program awards; and he was an Intelligence Community Postdoctoral Research Fellow, an NSF Graduate Research Fellow, a DoD NDSEG Fellow, and a NASA Texas Space
Grant Consortium Fellow.
\end{IEEEbiography}
\vfill




\end{document}